\documentclass[10pt,twocolumn,twoside]{IEEEtran}

\newcommand{\green}[1]{\textcolor{black!50!green}{#1}}

\usepackage{cite}\usepackage{hyperref}
\usepackage{amsmath,amssymb,amsfonts}
\usepackage{algorithmic}
\usepackage{graphicx}
\usepackage{textcomp}
\usepackage{amsmath,amssymb,bm,bbm,mathrsfs,amscd}
\usepackage{calc}
\usepackage{color}
\usepackage{amsthm}
\usepackage{dsfont}
\usepackage{graphicx}
\usepackage{tikz}
\usetikzlibrary{patterns}
\usepackage{amsmath}
\usepackage{amsfonts}
\usepackage{amssymb}
\usepackage{rotating}
\usepackage{mathtools}
\usepackage{color}
\usepackage{subfig}
\usepackage{enumerate,pgfplots}
\def\qed{\hfill \vrule height 7pt width 7pt depth 0pt\medskip}
\newtheorem{theorem}{Theorem}
\newtheorem{definition}{Definition}
\newtheorem{proposition}{Proposition}
\newtheorem{lemma}{Lemma}
\newtheorem{corollary}{Corollary}
\newtheorem{example}{Example}
\newtheorem{remark}{Remark}

\newcommand{\ba}{\begin{array}}
\newcommand{\ea}{\end{array}}

\newcommand{\be}{\begin{equation}}
\newcommand{\ee}{\end{equation}}
\newcommand{\abs}[1]{\lvert#1\rvert}

\newcommand{\mc}{\mathcal}

\newcommand{\ov}{\overline}
\newcommand{\ul}{\underline}
\newcommand{\nash}{\mc X_h^{*}}
\newcommand{\nashconsensus}{\mc X_h^{\bullet}}
\newcommand{\nashcoexistent}{\mc X_h^{\circ}}
\newcommand{\Z}{\mathbb{Z}}

\def\1{\boldsymbol{1}}

\newcommand{\R}{\mathbb{R}}

\renewcommand{\P}{\mathbb{P}}

\newcommand{\se}{\text{ if }}

\DeclareMathOperator*{\argmax}{argmax}

\DeclareMathOperator{\sgn}{sgn}

\def\Z{\mathbb{Z}}

\def\R{\mathbb{R}}

\def\P{\mathbb{P}}

\def\BibTeX{{\rm B\kern-.05em{\sc i\kern-.025em b}\kern-.08em
    T\kern-.1667em\lower.7ex\hbox{E}\kern-.125emX}}
\title{Robust Coordination of Linear Threshold Dynamics on Directed Weighted Networks}
\author{Laura~Arditti, 
Giacomo~Como,~\IEEEmembership{Member,~IEEE,}
        Fabio~Fagnani, and Martina Vanelli,~\IEEEmembership{Member,~IEEE}
\thanks{Some of the results in the paper appeared in preliminary form in \cite{Arditti.ea:2021}.}
\thanks{The authors are with the  Department of Mathematical Sciences ``G.L.~Lagrange,'' Politecnico di Torino, 10129 Torino, Italy  (e-mail: {\{laura.arditti;\,giacomo.como;\,fabio.fagnani;\,martina.vanelli\}@polito.it}). G. Como is also with the Department of Automatic Control, Lund University, 22100 Lund, Sweden.}
\thanks{This research was carried on within the framework of the MIUR-funded {\it Progetto di Eccellenza} of the {\it Dipartimento di Scienze Matematiche G.L.~Lagrange}, Politecnico di Torino, CUP: E11G18000350001. It received partial support from the MIUR  Research Project PRIN 2017 ``Advanced Network Control of Future Smart Grids'' (http://vectors.dieti.unina.it), and by the {\it Compagnia di San Paolo}. }
}

\begin{document}

\maketitle
\thispagestyle{empty}
\begin{abstract}
We study asynchronous dynamics in a network of interacting agents updating their binary states according to a time-varying threshold rule. Specifically, agents revise their state asynchronously by comparing the weighted average of the current states of their neighbors in the interaction network with possibly heterogeneous time-varying threshold values. Such thresholds are determined by an exogenous signal representing an external influence field modeling the different agents' biases towards one state with respect to the other one. We prove necessary and sufficient conditions for global stability of consensus equilibria, i.e., equilibria where all agents have the same state,  robustly with respect to the (constant or time-varying) external field. Our results apply to general weighted directed interaction networks and build on super-modularity properties of certain network coordination games whose best response dynamics coincide with the linear threshold dynamics. In particular, we introduce a novel notion of \emph{robust improvement paths} for such games and characterize conditions for their existence. 
\end{abstract}

\section{Introduction}\label{sec:introduction}
Robustness, meant as the ability of a system to maintain its performance under a range of different operating conditions,
is undoubtedly a fundamental issue that has long been studied in control \cite{Zhou.Doyle.Glover:1996}. While playing a key role in several domains, robustness and the related notion of resilience have lately become central in multi-agent and network systems, such as infrastructure systems \cite{Rinaldi.ea:2001,NAS-2012,Savla.ea:2014,Como:2017}, financial networks \cite{Eisenberg.Noe:2001,Haldane.May:2011,Acemoglu.ea:2015a,Massai.ea:2021}, as well as social and economic networks \cite{Acemoglu.ea:2011,Acemoglu.ea:2013,Acemoglu.ea:2012,Baqaee:2018}. In such contexts, robustness is typically presented as the capability of the system to react to localized perturbations by absorbing their effect  locally and preventing the global propagation of cascading failures that could prove detrimental for the whole system. A characteristic feature that has been recognized is that  the topology of the interconnection pattern is a key factor determining  the robustness or fragility of such network systems \cite{Acemoglu.ea:2015b,Como.ea:2013b,Como.Lovisari.Savla:2015,Sarkar.Roozbehani.Dahleh:2019,Savla.Shamma.Dahleh:2020}. 

In this paper, we focus on linear threshold dynamics (LTD), a prototypical family of nonlinear network systems first introduced  in \cite{mG:1978}  for fully mixed populations of agents and later extended in various directions  \cite{sM:2000, djW:2002, fVR:2007, EK:2010}. While LTD can be defined in different ways, their core structure consists of a set of agents identified with nodes of an interaction network that strategically change their binary state ($\pm1$) according to a threshold rule. Specifically, agents adopt state $+1$ if and only if the fraction of their neighbors in the interaction network that do so is greater than or equal to a certain exogenous threshold. Various studies of LTD models \cite{sM:2000, hA:2010, mL:2009, mL:2012, Moharrami:2016:EC, Rossi.ea:2019} have concerned topological conditions guaranteeing or preventing full contagion (i.e., convergence to configuration where all agents are in state $+1$) starting from an initial condition of relatively few agents in state $+1$. Most of these studies concern random networks of a specific type. A remarkable exception is \cite{sM:2000} that introduces the concept of cohesiveness of a subset of nodes in a network, through which one can in principle characterize the extent of a spreading phenomenon. 

In the literature, LTD models are assumed to be closed systems without explicit input or output signals. The basic challenge of this paper is to study LTD intrinsically equipped with an external field modeling a possibly node-specific influence from the external environment. As in the classical LTD models without external field the asymptotic outcomes are always consensus equilibria, our analysis concentrates on when a possibly time-varying  external filed can modify this behavior. Precisely, our results are of two types:
\begin{itemize}
\item robust stability results showing that the LTD converges to a consensus for every possibly time-varying external field taking values in a certain range;
\item control results showing that a suitable control signal is capable of preventing the system from reaching consensus by steering it to a different polarized configuration or by forcing persistent  oscillations.
\end{itemize}
Such behaviors will depend on the topology of the interaction network (building on suitable  generalizations of the concept of cohesiveness) and the constraints on the input signal.

LTD can be interpreted as the best response dynamics in a network game whereby agents choose strategically between two states and their payoff is an increasing function of the number of their neighbors choosing the same state. Such games are known as network coordination games and represent one of the most popular models to describe network systems with interactions of strategic complements type \cite{Milgrom.ROberts:1990,Vives:1990}. They find numerous applications in modeling  social and economic behaviors like the emergence of social norms and conventions or the adoption of new technologies \cite{Morris:2000,Young:2001,Young:1993,Young:2006,Montanari.Saberi:2010}. 

Optimal seeding and other intervention problems for network coordination games have been studied in \cite{Kempe.Kleinberg.Tardos:2003} and, in the more general setting of super-modular games, in \cite{Como.Durand.Fagnani:2020}. 
Our goal is different in this paper, as we are mainly interested in understanding the resilience of the system against external attacks. 
Recently, vulnerability of network coordination games against adversarial attacks has been investigated in \cite{Paarpon.ea:2021-TAC,Paarpon.ea:2021-TCNS}, while \cite{Jackson.Storms:2019} uses network coordination games as a micro-foundation for community structure in networks. 

Our analysis strongly relies on the interpretation of the LTD as the best response dynamics of a network coordination game. We then build on super-modularity of such games, i.e., the increasing difference property \cite{Topkins:1979,Topkins:1998}. Specifically, the  convenience for a player to switch from a state to an alternative state is monotone in the fraction of players in their neighborhood already playing the alternative state. 
 Such property continues to hold true under the influence of an external  field.  
 A variation of the external field modifies the threshold of the agents, in extreme cases transforming them into stubborn agents, i.e., agents whose best response is always the same state, regardless of her fellow agents' states. 

In particular, we study conditions under which a system converges to a consensus equilibrium, independently from the values taken by an external field. As it turns out, two conditions need to be satisfied for such robust stability property to hold true. The first condition, to be referred to as robust indecomposability, is a generalization of the lack of cohesive partitions \cite{Morris:2000} to parametrized families of heterogeneous network coordination games. It is equivalent (see Theorem \ref{theorem:robust-indecomposability}) to the lack of coexistent equilibria for any value of the external field within a certain range. On the other hand, the second condition guarantees that the external field is incapable of creating stubborn agents for both states. While the necessity of these two conditions for convergence to a consensus is quite intuitive, the proof of sufficiency is more involved and resides on the possibility to find best response paths for the game that are robust to modifications of the external field. This is achieved in Theorem \ref{theo:robustreachable} that is one of our main results and uses in a crucial way the super-modularity of the game. 

The rest of the paper is organized as follows. We report some basic notation in the remaining part of this section. In Section \ref{sec:main} we present the problem. We introduce the LTD with external field and the fundamental concept of indecomposability (Definition \ref{def:decomposition}). We then state two main results on the asymptotic of such model, Proposition \ref{prop:indecomposable-necessary} and Theorem \ref{theo:time-varying}, and we illustrate the outcomes through a number of examples and simulations. Section \ref{sec:coordination-games} is completely devoted to the analysis of network coordination games, especially the structure of the their set of Nash equilibria that play a crucial role in our study of the LTD. Section \ref{sec:robust-coordination} contains the core technical part of the paper. In particular, Theorem \ref{theo:robustreachable} contains robust reachability and stability results for network coordination games that are the fundamental ingredients to then prove Theorem \ref{theo:time-varying}. The paper is completed with a Section of conclusions and an Appendix containing some of the most technical proofs.

\subsection{Notation} 
For a finite set $\mc I$, we consider vector spaces $\R^{\mc I}$ 
equipped with the 
partial order 
$$	x \leq y \quad \Leftrightarrow \quad x_i \leq y_i, \quad \forall i \in \mc I \,.$$
We use the notation $x \lneq y$ when $x \leq y$ and $x_i < y_i$ for some $i$ in $\mc I$.
A function $f:\R^{\mc I}\to \R^{\mc J}$ is referred to as monotone nondecreasing (nonincreasing) if it preserves (reverses) the partial order $\leq$, i.e., if 
$f(x)\leq f(y)$ ($f(x)\geq f(y)$) for every $x\leq y$. 
For a vector $x$ in  $\R^{\mc I}$, $|x|$ in $\R^{\mc I}$ stands for the vector with entries $(|x|)_i=|x_i|$ for every $i$ in $\mc I$. The symbol $\1$ indicates a vector with all entries equal to $1$.

\section{Problem statement and main results} \label{sec:main}

Throughout the paper, we model networks as finite directed weighted graphs $\mc G=(\mc V, \mc E, W)$, with set of nodes $\mc V$, set of directed links $\mc E\subseteq\mc V\times\mc V$, and weight matrix $W$ in $\R_+^{\mc V\times\mc V}$, whose entries are such that $W_{ij}>0$ if and only if $(i,j)\in\mc E$. We do not allow for the presence of self-loops, equivalently, we assume that the weight matrix $W$ has zero diagonal. 
We refer to the network as undirected in the special case when the weight matrix $W=W'$ is symmetric, so that in particular there is a link $(i,j)$  directed from node $i$ to node $j$ in $\mc E$ if and only if there is also there is also the reverse link $(j,i)$ directed from node $j$ to node $i$ in $\mc E$.   

For a network $\mc G=(\mc V, \mc E, W)$ and a subset of nodes $\mc S\subseteq\mc V$, we denote by  
$$w_i^{\mc S} = \sum_{j\in \mc S} W_{ij}\,,$$
the $\mc S$-\textit{restricted out-degree} of a node $i$ in $\mc V$. 
In the special case when $\mc S=\mc V$ coincides with the whole node set, we simply refer to $w_i=w_i^{\mc V}$ as the \textit{out-degree} of a node $i$ in $\mc V$ and let $w=W\1$ be the vector of out-degrees.

The nodes of the network represent interacting agents. Every agent $i$ in $\mc V$ is endowed by a binary time-varying state $X_i(t)$.  A link $(i,j)$ in $\mc E$ is meant as directed from its tail node $i$ to its head node $j$, and represents a direct influence of agent $j$ on agent $i$, with its weight $W_{ij}$ to be interpreted as a measure of such influence. Let $\mc A=\{\pm1\}$ be the binary state set of each agent, and let $\mc X=\mc A^{\mc V}$ be the configuration space: a configuration $x$ in $\mc X$ is a vector whose entries $x_i$ represent the states of the single agents. The constant vectors $x=\pm\1$ will be referred to as \emph{consensus} configurations. On the other hand, we shall refer to every $x$ in $\mc X\setminus\{\pm\1\}$ as a \emph{co-existent } configuration.

We consider asynchronous time-varying (ATV) \emph{linear threshold dynamics} (LTD) on a network $\mc G=(\mc V,\mc E,W)$, whereby agents $i$ in $\mc V$ update their binary state $X_i(t)$ in $\mc A$ as described below. For a nonempty $\mc H\subseteq\mc R^{\mc V}$, let $h(t)$ in $\mc H$ for $t\ge0$ be an exogenous signal modeling a time-varying external field. Let every agent $i$ in $\mc V$ be equipped with an independent rate-$1$ Poisson clock.\footnote{The assumption that all Poisson clocks have rate $1$ is made merely for the sake of simplicity of the exposition. In fact, it is not hard to show that all results in the paper continue to hold true as stated in the more general setting where every agent $i$'s Poisson clock has rate $\lambda_i>0$.} If agent $i$'s clock ticks at some time $t\ge0$,\footnote{Observe that, with probability $1$, no two agents' clocks will ever tick at the same time $t$. } then agent $i$ modifies her current state $X_i(t^-)$ into a new state $X_i(t)$ such that 
\be\label{dynamics}X_i(t)=\left\{\ba{lcl}+1&\se&\sum_jW_{ij}X_j(t)+h_i(t)>0\\[3pt]X_i(t^-)&\se&\sum_jW_{ij}X_j(t)+h_i(t)=0\\[3pt]-1&\se&\sum_jW_{ij}X_j(t)+h_i(t)<0\,.\ea\right.\ee
The update rule above can be rewritten in the following equivalent way. 
For $i$ in $\mc V$ and time $t\ge0$, let 
\be\label{eq:ri} r_i(t)=\frac12-\frac{h_i(t)}{2w_i}\,,\ee
be a time-varying threshold for agent $i$. 
Also, for a configuration $x$ in $\mc X$, let 
\be\label{wi+wi-}w_i^-(x)=\sum_{j:x_j=-1}W_{ij}\,,\qquad w_i^+(x)=\sum_{j:x_j=1}W_{ij}\,,\ee
be the aggregate weight of links pointing from agent $i$ to agents in state $-1$ and, respectively, to those in state $+1$. 
Then,  \eqref{dynamics} is equivalent to 
\be\label{dynamics-bis}X_i(t)=\left\{\ba{lcl}
+1&\se&w_i^+(X(t^-))>w_ir_i(t)\\[3pt]
X_i(t^-)&\se&w_i^+(X(t^-))=w_ir_i(t)\\[3pt]
-1&\se&w_i^+(X(t^-))<w_ir_i(t)\,,\ea\right.\ee
i.e., if agent $i$ gets activated at time $t\ge0$, then: (a) she updates her state $X_i(t)$ to $+1$ if the weighted fraction $w_i^+(X(t^-))/{w_i}$ of her out-neighbors currently in state $+1$ is above the time-varying threshold $r_i(t)$ (equivalently, if the weighted fraction $w_i^-(X(t^-))/{w_i}$ of her out-neighbors in state $-1$ is below the complementary threshold $1-r_i(t)$); (b)  she updates her state $X_i(t)$ to $-1$ if $w_i^+(X(t^-))/{w_i}$ is below  $r_i(t)$
; or (c) she keeps her current state $X_i(t)=X_i(t^-)$ if $w_i^-(X(t^-))/{w_i}=r_i(t)$.


If we stack the agents' states in a vector $X(t)$ in $\mc X$, then $X(t)$ is a continuous-time inhomogeneous Markov chain on the configuration space $\mc X$. 
In the rest of the paper, we shall focus on the asymptotic behavior of the ATV-LTD $X(t)$ on a network $\mc G=(\mc V,\mc E,W)$ with external field $h(t)$, as described above. 
We shall refer to a configuration $x^*$ in $\mc X$ as \emph{absorbing} for the ATV-LTD, if  $X(t^*)=x^*$ for some $t^*\ge0$ implies that $X(t)=x^*$ for every $t\ge t^*$. 

Specifically, we shall determine necessary and sufficient conditions for almost sure convergence (i.e., convergence with probability one) to a consensus configuration. In particular, our main result concerns robust convergence to consensus for ATV-LTD on a network $\mc G$ when the external field $h(t)$ is an arbitrary (unknown) signal whose range is a hyper-rectangle in the form
 \be\label{H}\mc H=\{h:\,h^-\le h\le h^+\}=\prod\limits_{i\in\mc V}\left[h^-_i,h^+_i\right]\,,\ee
for two (known) vectors $h^-$ and $h^+$ in $\R^{\mc V}$ such that $h^-\le h^+$. 

The conditions for robust almost sure convergence to consensus of the ATV-LTD will be determined in terms of graph-theoretic properties of the network $\mc G$. In particular, we have the following definition. 
\begin{definition}\label{def:decomposition} 
Let $h^-$ and $h^+$ in $\mc R^{\mc V}$ be two vectors such that $h^-\le h^+$. 
Then, a network $\mc G=(\mc V,\mc E,W)$ is $(h^-,h^+)$-\emph{indecomposable} if
for every nontrivial binary partition of the node set 
\be\label{eq:binary-partition}\mc V = \mc V^+ \cup \mc V^-\,,\quad \mc V^-\cap\mc V^+=\emptyset\,,\quad \mc V^+\ne\emptyset\,,\quad \mc V^-\ne\emptyset\,,\ee  
there exist  $s$ in $\{-,+\}$ and a node $i$ in $\mc V^s$ such that
	\begin{equation}\label{eq:robustly_indec}
	w_i^{s} + sh_i^{s}< w_i^{-s}\,,
	\end{equation}
	where
$w^s_i=w^{\mc V^{s}}_i$. 
In the special case when $h^-=h^+=h$, we shall more briefly refer to the network $\mc G$ as $h$-indecomposable. 
\end{definition}
\begin{figure}\begin{center}
\includegraphics[width=5cm]{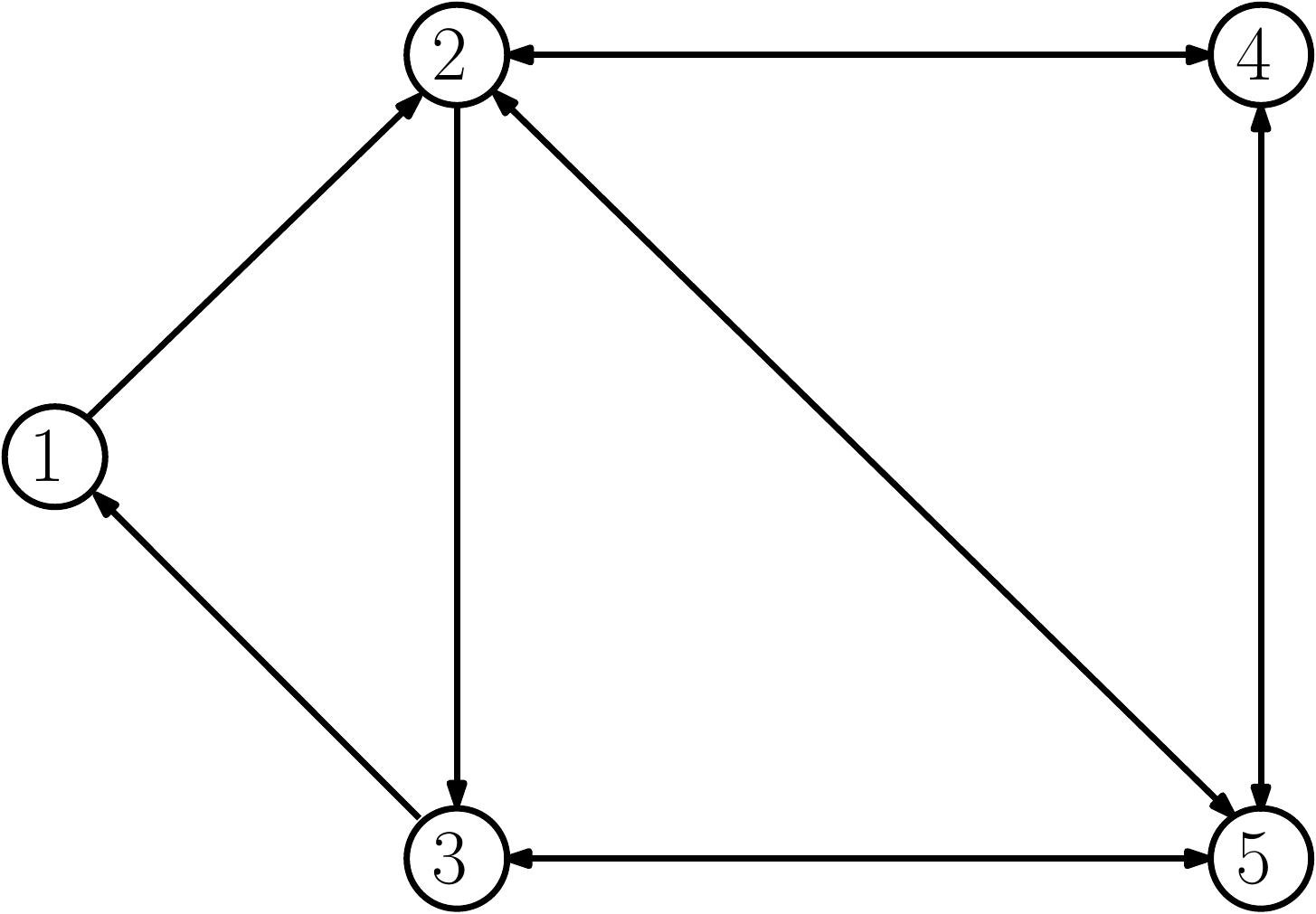} 
\end{center}
\caption{\label{fig:graph-5}}\end{figure}

The following example illustrates the notion of indecomposability introduced in Definition \ref{def:decomposition} above in a simple case. 
\begin{example}\label{ex:graph-5}
Consider the network $\mc G=(\mc V,\mc E,W)$ displayed in Figure \ref{fig:graph-5}, with set of nodes $\mc V=\{1,2,3,4,5\}$ and 
weight matrix 
$$W=\left(\ba{ccccc}
0&1&0&0&0\\
0&0&1&1&1\\
1&0&0&0&1\\
0&1&0&0&1\\
0&1&1&1&0
\ea\right)\,.$$
The out-degree vector is then $w=W\1=(1,3,2,2,3)\,.$   

(i) First, we verify that $\mc G$ is $0$-indecomposable. To see it, first notice that if, for some $s$ in $\{-,+\}$, either $\mc V_s=\{i\}$ for some $i$ in $\mc V$ or $\mc V_s=\{i,j\}$ for some $i\ne j$ in $\mc V$ such that $(i,j)\notin\mc E$, then we have $w_i^s=0<1\le w^{-s}_i$, so that \eqref{eq:robustly_indec} is satisfied. 
It is then sufficient to consider binary partitions as in \eqref{eq:binary-partition} where,  $\mc V_s=\{i,j\}$ for some $s$ in $\{-,+\}$ and  $i\ne j$ in $\mc V$ such that both $(i,j)$ and $(j,i)$ belong to $\mc E$. This leaves us with four possibilities, corresponding to the four undirected links in the network: (a) for $\mc V_s=\{2,4\}$, we have that $5\in\mc V_{-s}$ and $w_5^s=2>1=w^{-s}_5$ so that \eqref{eq:robustly_indec} is satisfied;  for $\mc V_s=\{2,5\}$, we have that $4\in\mc V_{-s}$ and $w_4^s=2>0=w^{-s}_4$ so that \eqref{eq:robustly_indec} is satisfied; (c) for  both $\mc V_s=\{3,5\}$ and $\mc V_s=\{4,5\}$, we have that $2\in\mc V_{-s}$ and $w_2^s=2>1=w^{-s}_2$ so that \eqref{eq:robustly_indec} is satisfied. 

(ii) Second, we verify that $\mc G$ is not $\delta^1$-indecomposable. Indeed, let us fix $\mc V_-=\{2,4,5\}$ and $\mc V_+=\{1,3\}$. Then, 
$w_1^++1=1=w_1^-$, $w_2^-=2>1=w_2^+$, $w_3^+=1=w_3^-$, $w_4^-=2>0=w_4^+$, and $w_5^-=2>1=w_5^+$, so that \eqref{eq:robustly_indec} is violated by every $i$ in $\mc V_s$ and $s$ in $\{-,+\}$. 

(iii) Now, we show that   $\mc G$ is not $(h^-,h^+)$-indecomposable for $h^-=(0,-1,0,0,0)$ and $h^-=(0,0,0,0,1)$. Indeed, let us fix $\mc V_-=\{1,2,3\}$ and $\mc V_+=\{4,5\}$. Then, we have that 
$w_1^{-} -h_1^{-}=1>0=w_1^{+}$, $w_2^{-} -h_2^{-}=1+1=2=w_2^{+}$, $w_3^{-} -h_3^{-}=1= w_3^{+}$, $w_4^{+} + h_4^{+}=1=w_4^{-}$, and also $w_5^{+} + h_5^{+}=1+1=2=w_5^{-}$, so that \eqref{eq:robustly_indec} is violated by every $i$ in $\mc V_s$ and $s$ in $\{-,+\}$. Notice that, in contrast, it can be verified that $\mc G$ is both $h^-$-indecomposable and $h^+$-indecomposable in this case.

(iv) Finally, let $h^-=0$ and $h^+=(0,2,0,0,2)$. We now show that   $\mc G$ is $(h^-,h^+)$-indecomposable. To verify that, first notice that if $\mc V_s=\{i\}$ for some  $s$ in $\{-,+\}$ and $i$ in $\mc V$, then $w_i^s+sh_i^s=sh_i^s< w_i=w^{-s}_i$, so that \eqref{eq:robustly_indec} is satisfied. Similarly, if $1\in\mc V_s$ and $2\in\mc V_{-s}$, then $w_1^{s} + sh_1^{s}=0<1=w_1^{-s}$. 
Moreover, if $\{2,5\}\subseteq\mc V_s$ for some   $s$ in $\{-,+\}$, then  \eqref{eq:robustly_indec} is satisfied by every $i$ in $\mc V_{-s}\cap\{1,4\}$. 
This leaves us with four possibilities: (a) $\mc V_-=\{1,2,3\}$ and $\mc V_{+}=\{4,5\}$; (b)  $\mc V_-=\{1,2,4\}$ and $\mc V_{+}=\{3,5\}$; (c) $\mc V_-=\{4,5\}$ and $\mc V_{+}=\{1,2,3\}$; (d) $\mc V_-=\{3,5\}$ and $\mc V_{+}=\{1,2,4\}$. In both cases (a) and (b), we have $w_2^{-} -h_2^{-}=1<2=w_2^{+}$ so that \eqref{eq:robustly_indec} is satisfied, whereas in both cases (c) and (d), we have $w_5^{-} -h_5^{-}=1<2=w_5^{+}$ so that \eqref{eq:robustly_indec} is satisfied. 
Therefore, $\mc G$ is $(h^-,h^+)$-indecomposable.
%
%
\end{example}


Our first result stated below shows that $(h^-,h^+)$-indecomposability of the network $\mc G$ is indeed a necessary condition for robust convergence to a consensus configuration of ATV-LTD when the external field $h(t)$ is a arbitrary signal whose range is the hyper-rectangle \eqref{H}. 

\begin{proposition}\label{prop:indecomposable-necessary}
Let $\mc G=(\mc V,\mc E,W)$ be a network. For two vectors $h^-$ and $h^+$ in $\mc R^{\mc V}$ such that  $h^-\le h^+$, let $\mc H$ be as in \eqref{H}.
If $\mc G$ is not $(h^-,h^+)$-indecomposable, then there exist $h^*$ in $\mc H$ and a coexistent configuration $x^*$ in $\mc X\setminus\{\pm\1\}$ such that $x^*$ is an absorbing configuration for the ATV-LTD on $\mc G$ with constant external field $h(t)=h^*$.  
\end{proposition}
\begin{proof} See Appendix \ref{sec:proof-prop-indecomposable-necessary}.
\end{proof}

\begin{example}\label{ex:graph-5-bis}
Consider the network $\mc G$ shown in Figure \ref{fig:graph-5}  and let $h^-=(0,-1,0,0,0)$ and $h^-=(0,0,0,0,1)$. As verified in Example \ref{ex:graph-5},  $\mc G$ is not $(h^-,h^+)$-indecomposable, so that Proposition \ref{prop:indecomposable-necessary} implies the existence of a vector $h^*$ such that $h^-\le h^*\le h^+$ and of a coexistent configuration $x^*$ in $\mc X\setminus\{\pm\1\}$ such that $x^*$ is a fixed point for the LTD on $\mc G$ with constant external field $h(t)=h^*$.  
Specifically, in this case we can take $h^*=(0,-1,0,0,1)$ and $x^*=(-1,-1,-1,+1,+1)$.  
\end{example}

While Proposition \ref{prop:indecomposable-necessary} states that, if the network $\mc G$ is not $(h^-,h^+)$-indecomposable, robust convergence to a consensus configuration is not ensured for the ATV-LTD on $\mc G$,  the following result establishes necessary and sufficient conditions for robust convergence to a consensus configuration when the network $\mc G$ is $(h^-,h^+)$-indecomposable. 

\begin{theorem}\label{theo:time-varying} 
Let $\mc G=(\mc V,\mc E,W)$ be a network. For two vectors $h^-$ and $h^+$ in $\mc R^{\mc V}$ such that  $h^-\le h^+$, let $\mc H$ be as in \eqref{H}.
If $\mc G$ is $(h^-,h^+)$-indecomposable, then the ATV-LTD on $\mc G$ with any external field $h(t)\in\mc H$ for $t\ge0$ is such that, for every initial configuration $X(0)$ in $\mc X$,  with probability $1$ there exists $t^*\ge0$ such that
\begin{enumerate}
\item[(i)]
\be\label{reachable}X(t^*)\in\{\pm\1\};\ee 
\end{enumerate}
and 
\begin{enumerate}
\item[(ii)] if $w>-h^-$ and $w>h^+$, then 
\be\label{eq:caseii}X(t)\in\{\pm\1\}\,,\qquad\forall t\geq t^*\,;\ee
\item[(iii)] 
if 
$w>-ah^{-a}$ and  $w\ngeq a h^{-a}$ 
for some $a=\pm1$, then
\be\label{eq:casei}X(t)=a\1\,,\qquad\forall t\geq t^*\,.\ee

\end{enumerate}
Moreover: 
\begin{enumerate}
\item[(iv)] if 
$w\ngeq - h^{-}$ and $w\ngeq h^{+}$, 
then there exists a signal $h(t)\in\mc H$ for  $t\ge0$ such that, for every initial configuration $X(0)$ in $\mc X$, with probability $1$ the ATV-LTD on $\mc G$  with external field $h(t)$ visits both consensus configurations $+\1$ and $-\1$ infinitely often. 
\end{enumerate}
\end{theorem}
\begin{proof} See Section \ref{sec:coordination-dynamics}. \end{proof}

Theorem  \ref{theo:time-varying} (i) is to be interpreted as a converse to Proposition \ref{prop:indecomposable-necessary}, as it states that, if the network $\mc G$ is $(h^-,h^+)$-indecomposable, then the set $\{\pm\1\}$ is reached in finite time with probability $1$ from every initial configuration $X(0)$. Theorem  \ref{theo:time-varying} (i) and Proposition \ref{prop:indecomposable-necessary} together then guarantee that $(h^-,h^+)$-indecomposability of the network $\mc G$ is a necessary and sufficient condition for global robust reachability of the set of consensus configurations for the ATV-LTD. 

Moreover, because of the form of the update rule \eqref{dynamics} the condition $w>-h^-$ implies that the consensus configuration $+\1$ is absorbing. Symmetrically, the condition $w>h^-$ implies that the consensus configuration $-\1$ is absorbing. Hence, point (ii) of Theorem  \ref{theo:time-varying} states that if, both $w>-h^-$ and $w>h^-$, then 
for every initial configuration $X(0)$ in $\mc X$, the ATV-LTD gets absorbed in finite time in one of these two consensus configurations. In fact, the probability with which $X(t)$ is absorbed in the consensus configuration $+\1$ rather than in $-\1$ will depend on the initial configuration $X(0)$, the network $\mc G$, and the particular external field $h(t)$. 

On the other hand, let us consider the case $a=+1$ in point (iii) of Theorem \ref{theo:time-varying} (the case $a=-1$ being completely symmetrical). Then, as discussed above, the condition $w>-h^{-}$ ensures that the consensus configuration $x^*=+\1$ is absorbing. On the other hand, the condition $w\ngeq h^{-}$ is equivalent to the existence of some agent $i$ in $\mc V$ such that $w_i<h^-_i$. From the form of the update rule \eqref{dynamics}, we then deduce that any such agent $i$ will switch her action to $+1$ the first time she gets activated and stick to $X_i(t)=+1$ ever after. Point (iii) of Theorem  \ref{theo:time-varying} then states that with probability $1$ all other agents will follow such agent $i$ and switch to state $+1$ in a cascade until the absorbing consensus configuration $+\1$ is reached in finite time. 

Finally, in contrast to points (i)-(iii), point (iv) of Theorem \ref{theo:time-varying} does not describe a robust behavior. Rather, the two conditions $w\ngeq - h^{-}$ and $w\ngeq h^{+}$ ensure that there exist two (not necessarily distinct) agents $i$ and $j$ in $\mc V$ such that $w_i<-h_i^-$ and $w_j<h^+_j$, respectively. This implies that agent $i$ will always switch to $-1$ the first time she gets activated under the external field $h^-$, while agent $j$ will always switch to $+1$ the first time she gets activated under the external field $h^+$. Point (ii) of  Theorem  \ref{theo:time-varying} then states that, by manoeuvring the external field $h(t)$ within its range $\mc H$, one is then able to make the system oscillate infinitely often between the two consensus configurations.

%

The proof of Theorem \ref{theo:time-varying} is one of the main contributions of this paper. The key technical challenges are twofold: on the one hand, we are considering LTD with time-varying external field $h(t)$ and seeking robustness results with respect to $h(t)$, on the other hand, considering weighted directed networks prevents one from appealing to potential games arguments. 
We will address these challenging by first focusing on the special case of LTD with constant external field and studying it from a super-modular game theory perspective in Section \ref{sec:coordination-games}. We will then introduce and characterize the key notion of robust improvement path in Section \ref{sec:robust-coordination} and finally apply it to prove Theorem  \ref{theo:time-varying} in Section \ref{sec:coordination-dynamics}.

\begin{figure}
\begin{center}
\includegraphics[width=7cm]{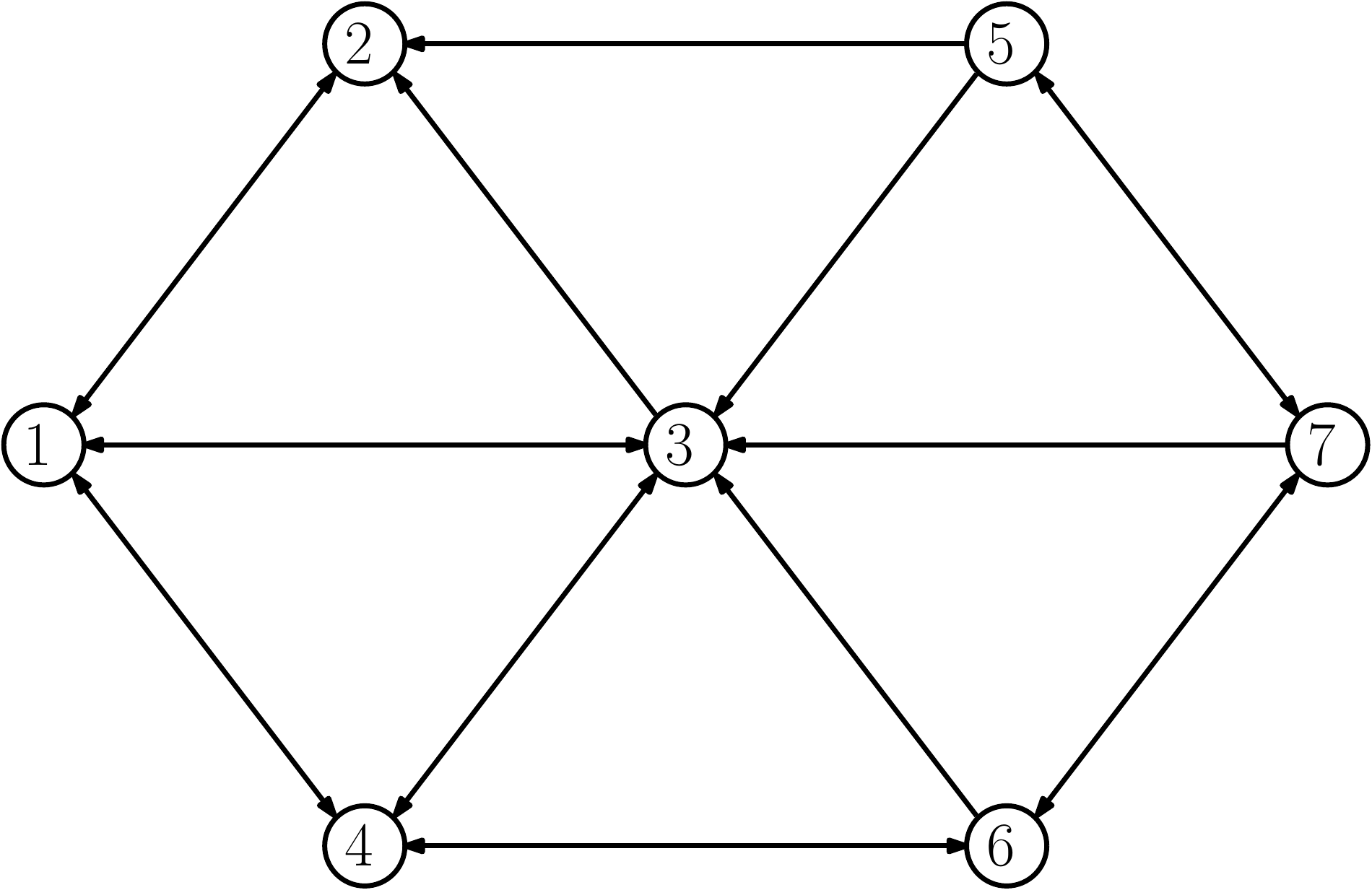}
\end{center}
\caption{\label{fig:graph-7}}
\end{figure}

\begin{figure}
	\begin{center}
		\includegraphics[width=9cm]{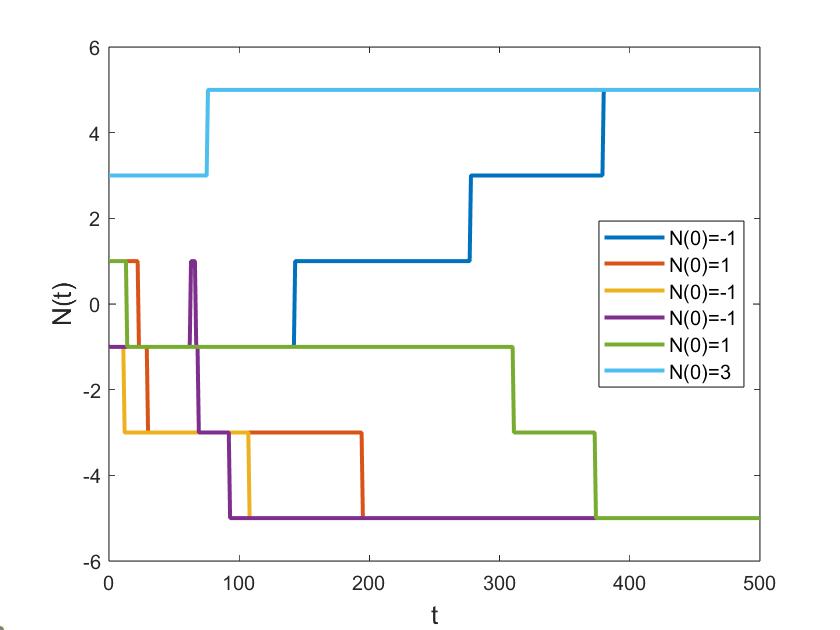}
		\includegraphics[width=9cm]{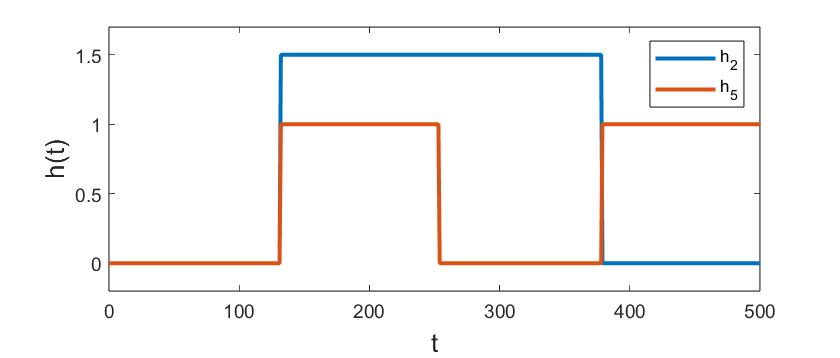}
	\end{center}
	\caption{In the upper panel, dynamics of $N(t)=\sum_iX_i(t)$ for the network in Figure \ref{fig:graph-5} and $h(t)=(0,h_2(t), 0,0,h_5(t))$ with $h_2$ and $h_5$ as in the lower panel. ATV-LTD get absorbed in consensus configurations for different initial conditions (see Example \ref{ex:graph-5-ter}). \label{fig:ltd_ex3}}
\end{figure}
\begin{example}\label{ex:graph-5-ter}
	Consider once again the network $\mc G$ shown in in Figure \ref{fig:graph-5} 
	%
	and  let $h^+=(0,2,0,0,2)$. As is shown in Example \ref{ex:graph-5}, $\mc G$ is $(0,h^+)$-indecomposable. Hence, since $w>-h^-$ and $w>h^+$, Theorem \ref{theo:time-varying} (ii) implies that the ATV-LTD on $\mc G$ with any external field $0\le h(t)\le h^+$ for $t\ge0$ gets absorbed with probability $1$ in finite time in a consensus configuration.
	
	In Figure \ref{fig:ltd_ex3}, we simulated the dynamics of $N(t)=\sum_iX_i(t)$ for different initial conditions when the external field is $h(t)=(0,h_2(t), 0,0,h_5(t))$ with $h_2$ and $h_5$ as in Figure \ref{fig:ltd_ex3} (notice that $0\le h(t)\le h^+$). ATV-LTD dynamics get absorbed in consensus configurations. 
\end{example}

\begin{figure}
	\begin{center}
		\includegraphics[width=9cm]{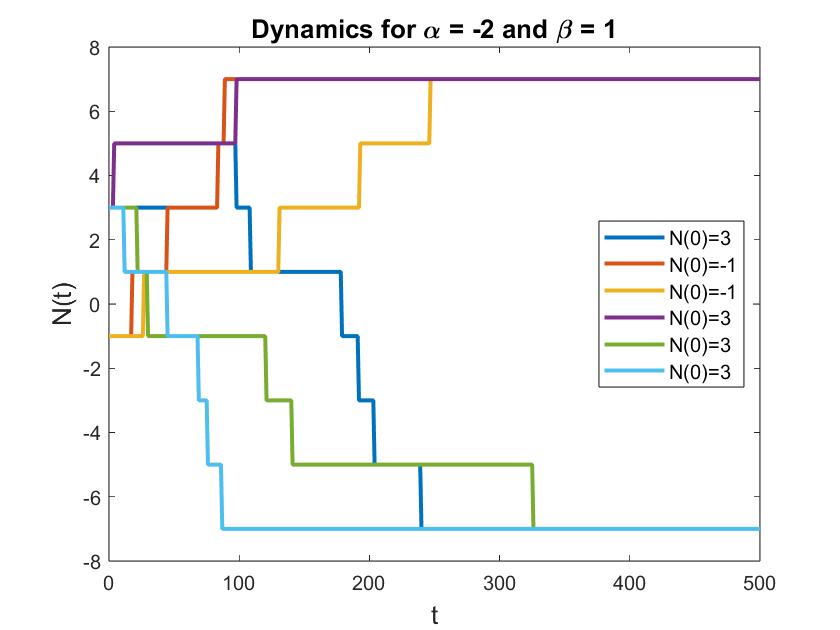}
		\includegraphics[width=9cm]{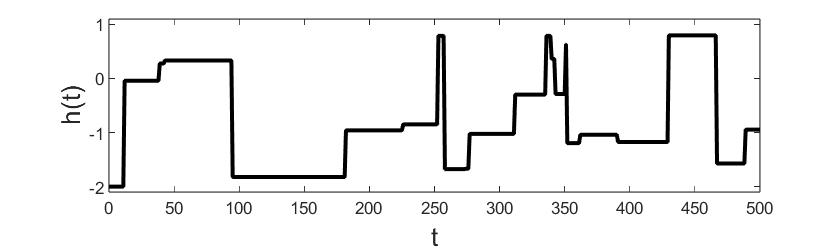}
	\end{center}
	\caption{In the upper panel, dynamics of $N(t)=\sum_iX_i(t)$ for the network in Figure \ref{fig:graph-7} and $h(t)=(h_1(t),0, 0,0,0)$ with $h_1$ as in the lower panel. ATV-LTD get absorbed in consensus configurations for different initial conditions (see Example \ref{ex:graph-7} with $\alpha=-2$ and $\beta=1$).  \label{fig:ex4_case1}}
\end{figure}
\begin{figure}
	\begin{center}
		\includegraphics[width=9cm]{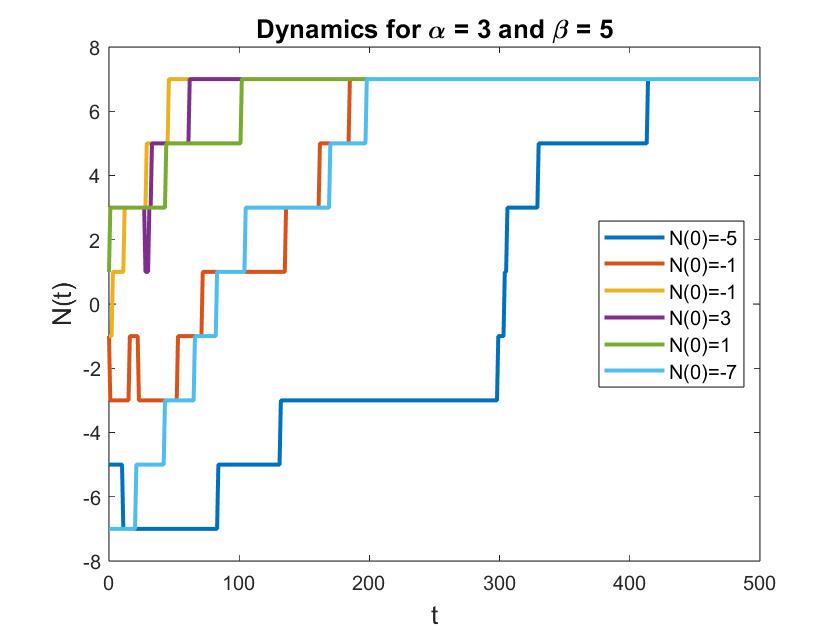}
		\includegraphics[width=9cm]{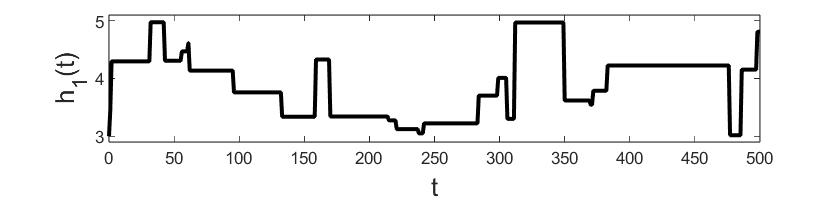}
	\end{center}
	\caption{In the upper panel, dynamics of $N(t)=\sum_iX_i(t)$ for the network in Figure \ref{fig:graph-7} and $h(t)=(h_1(t),0, 0,0,0)$ with $h_1$ as in the lower panel. ATV-LTD get absorbed in the consensus configuration $x^*=+\1$ for different initial conditions (see Example \ref{ex:graph-7} with  $\alpha=3$ and $\beta=5$). \label{fig:ex4_case2}}
\end{figure}		
\begin{figure}
	\begin{center}		
		\includegraphics[width=9cm]{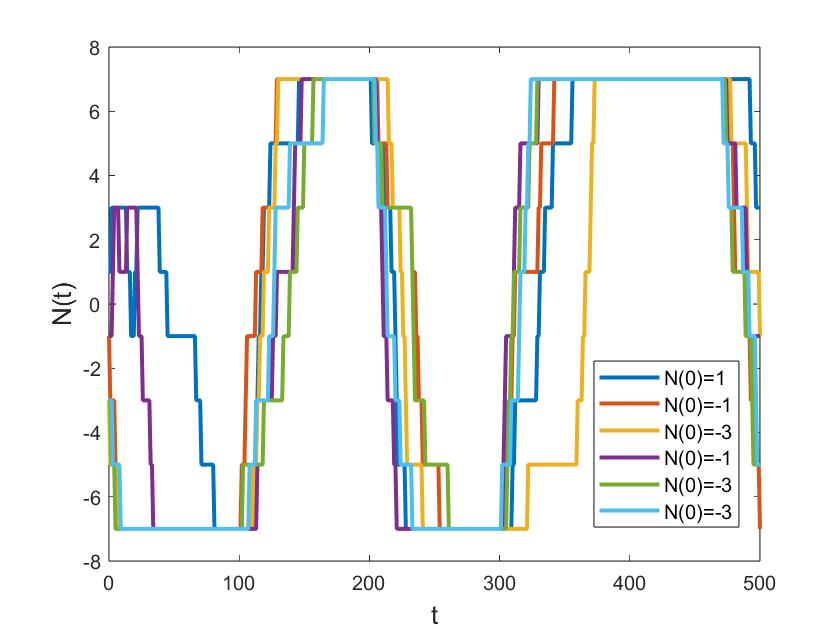}
		\includegraphics[width=9cm]{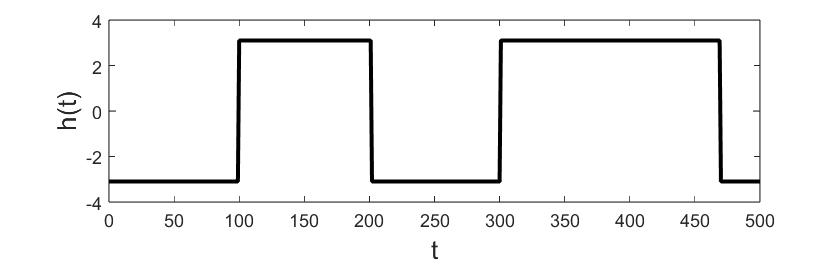}
	\end{center}
	\caption{In the upper panel, dynamics of $N(t)=\sum_iX_i(t)$ for the network in Figure \ref{fig:graph-7} and $h(t)=(h_1(t),0, 0,0,0)$ with $h_1$ as in the lower panel. ATV-LTD fluctuate for different initial conditions (see Example \ref{ex:graph-7} with $\alpha=-3.1$ and $\beta=3.1$). \label{fig:ex4_case3}}
\end{figure}

\begin{example}\label{ex:graph-7}
Consider the network $\mc G$ shown in in Figure \ref{fig:graph-7}, 
with node set $\mc V=\{1,\ldots,7\}$ and out-degree vector $w=(3,1,3,3,3,3,3)$. 
Let $$h^-=(\alpha,0,\ldots,0)\,,\qquad h^+=(\beta,0,\ldots,0)\,,$$  for $\alpha\le\beta$ in $\R$.
Then, for every nontrivial binary partition as in \eqref{eq:binary-partition}, let $s$ in $\{-,+\}$ be such that $1\in\mc V_{-s}$. If $2\in\mc V_{s}$, then $w_2^{s} + sh_2^{s}=0<1= w_2^{-s}$, so that \eqref{eq:robustly_indec} is satisfied. On the other hand, for $j=2,\ldots,6$, if $\{1,\ldots,j\}\subseteq\mc V_{-s}$ and $j+1\in\mc V_s$, then $w_{j+1}^{s} + sh_{j+1}^{s}\le 1<2\le w_{j+1}^{-s}$, so that \eqref{eq:robustly_indec} is satisfied. This proves that $\mc G$ is $(h^-,h^+)$-indecomposable for every $\alpha\le\beta$. Now consider three different cases.  

If $-3<\alpha\le\beta<3$, so that $w>-h^-$ and $w>h^+$, then Theorem \ref{theo:time-varying} (ii) ensures that, with probability $1$,  $X(t)$ gets absorbed in finite time in one of the two consensus configurations (see Figure \ref{fig:ex4_case1}). 

On the other hand, if $3<\alpha\le\beta$, so that $w>-h^-$ and $w\ngeq h^-$, then Theorem \ref{theo:time-varying} (iii) ensures that, with probability $1$, $X(t)$ gets absorbed in finite time in the consensus configuration $x^*=+\1$ (see Figure \ref{fig:ex4_case2}). 

Finally, if $\alpha<-3$ and $\beta>3$, so that $w\ngeq - h^{-}$ and $w\ngeq h^{+}$ then Theorem \ref{theo:time-varying} (iv) ensures that there exists a time-varying signal $h^-\le h(t)\le h^+$  such that, with probability $1$,  $X(t)$ fluctuates forever between the two consensus configurations visiting both of them infinitely many times (see Figure \ref{fig:ex4_case3}). 

 \end{example}

\section{Network coordination games}\label{sec:coordination-games}
In this section, we study the special case of LTD with constant external field. 
As we shall see, in this special case, the LTD can be reinterpreted as the asynchronous best response dynamics of a network coordination game and its fixed points correspond to the (pure strategy Nash) equilibria of such game. We shall provide a full characterization of such equilibria and of the asymptotic behavior of the corresponding LTD in terms of graph-theoretic properties of the network.

\subsection{Game-theoretic notions}\label{sec:game-theoretic}
We start by introducing some general game-theoretic notions. We consider strategic form games with binary action set $\mc A=\{\pm1\}$ played by rational agents $i$ in $\mc V$ whose aim is to maximize their own utility function $u_i:\mathcal{X}\to \mathbb{R}\,.$  As customary, for a strategy profile $x$ in $\mc X$ and a player $i$ in $\mc V$, we let $x_{-i}$ in  $\mathcal{X}_{-i}= \mc A^{\mathcal{V}\setminus\{i\}} $ stand for the strategy profile of all players except for player $i$. 
We shall then use the common abuse of notation 
$u_i(x)=u_i(x_i,x_{-i})$
for the utility perceived by player $i$ in configuration $x$. 
The \textit{best response} correspondence for a player $i$  in $\mc V$ is  then
$$\mc B_i(x_{-i})= \argmax_{x_i \in\mc A}u_i(x_i,x_{-i}) \,.$$
An action $a$ in $\mc A$ is \emph{dominant} (\emph{strictly dominant}) for a player $i$ if 
$a\in\mc B_i(x_{-i})$ ($\mc B_i(x_{-i})=\{a\}$) for every $ x_{-i}$ in $\mc X_{-i}$. 
A player $i$ having a strictly dominant action $a$ is referred to as an \emph{$a$-stubborn} agent. 
A (pure strategy Nash) \emph{equilibrium}  is a configuration $x^*$ in $\mc X$ such that 
$$x^*_i\in\mc B_i(x^*_{-i})\,,\qquad\forall i\in\mc V\,.$$
The set of equilibria of a game is denoted by $\mc X^*$. 
An equilibrium $x^*$ is \emph{strict} if 
$\mc B_i(x^*_{-i})=\{x^*_i\}$ for every $i$ in $\mc V$. 

We shall consider 
revision protocols whereby only one player modifies her action at a time and such modification never decreases the player's utility. A number of related concepts are reported below.
\begin{definition}
For $l\ge0$, a length-$l$ \emph{admissible path} from $x$ in $\mc X$ to $y$ in $\mc X$ is an $(l+1)$-tuple 
$(x^{(0)},x^{(1)},\ldots x^{(l)})$ in $\mc X^{l+1}$ 
of strategy profiles such that 
$x^{(0)}=x$, $x^{(l)}=y$, 
and, for $k=1,2,\ldots, l$, there exists a player $i_k$ in $\mc V$ such that  
\be\label{x-i_k=y-i_k}x^{(k)}_{-i_k}=x^{(k-1)}_{-i_k}\,, \qquad x_{i_k}^{(k)}\ne x_{i_k}^{(k-1)}\,.\ee
A length-$l$ admissible path as above is: 
\begin{itemize} 
\item \emph{monotone}  if $x^{(0)}\lneq x^{(1)}\lneq \cdots\lneq x^{(l)}$;
\item \emph{anti-monotone}  if $x^{(0)}\gneq x^{(1)}\gneq \cdots\gneq  x^{(l)}$.
%
%
\item an \emph{improvement path} (\emph{I-path}) if
\be\label{eq:improvement}u_{i_k}(x^{(k)})>u_{i_k}(x^{(k-1)})\,,\qquad k=1,2,\ldots, l\,;\ee
\item a \emph{best response path} (\emph{BR-path}) if
\be\label{eq:wk-improvement}u_{i_k}(x^{(k)})\geq u_{i_k}(x^{(k-1)})\,,\qquad k=1,2,\ldots, l\,.\ee
%
%
\end{itemize}
\end{definition}

In other words, a length-$l$ admissible path is an $(l+1)$-tuple of strategy profiles such that two consecutive strategy profiles differ just for the action of a single active player. Notice that an admissible path from a configuration $x$ to another configuration $y$ is completely determined by the sequence of active players $(i_1,i_2,\dots , i_l)$. In monotone paths, players can modify their actions from $-1$ to $+1$ only, conversely only modifications from $+1$ to $-1$ are allowed in anti-monotone paths.
Notice that, by convention, every I-path is also a BR-path (but not vice versa) and that the singleton $(x)$ is to be considered as a length-$0$ I-path from $x$ to $x$ with an empty sequence of active players. 
%
%
%

\begin{definition}
For $\alpha$ in $\{\text{I},\text{BR}\}$, a subset of strategy profiles $\mc Y\subseteq\mc X$ is:
\begin{itemize}
\item $\alpha$-\emph{reachable}  from strategy $x$ in $\mc X$ if there exists an $\alpha$-path from $x$ to some strategy profile $y$ in $\mc Y$;  
\item \emph{globally} $\alpha$-reachable if it is $\alpha$-reachable from every configuration $x$ in $\mc X$; 
\item $\alpha$-\emph{invariant} if there are no $\alpha$-paths from any $y$ in $\mc Y$ to any $z$ in $\mc X\setminus\mc Y$. 
\item \emph{globally $\alpha$-stable} if it is globally $\alpha$-reachable and $\alpha$-invariant. 
\end{itemize}
\end{definition}

To indicate that a configuration $y$ is I-reachable (BR-reachable) from a configuration $x$, we will use the notation 
$x\to y$ ($x\rightharpoonup y$).
If $y$ is reachable from $x$ by an  monotone or anti-monotone I-path (BR-path), we will use the notation
$$x\uparrow y,\quad x\downarrow y,\quad x\upharpoonleft y,\quad x\downharpoonright y\,,$$
respectively. Observe that all these relations are reflexive and transitive.

%
%

\subsection{Binary super-modular games}\label{sec:supermodular}
Super-modular games are an important class of games and have received a considerable amount of attention in the literature \cite{Topkins:1979,Milgrom.ROberts:1990,Vives:1990,Topkins:1998}. Below we recall the definition of super-modular games and report a number of their fundamental properties. 
We take advantage of the binary action setting considered here, while at the same time we recall that the theory of super-modular games can be developed in more general lattice action sets.

Let us first introduce some notation. For two vectors $x$ and $y$ in $\R^{\mc I}$, the entry-wise supremum (or least upper bound) $x\vee y$ in $\R^{\mc I}$ and infimum (or greatest lower bound)  $x\wedge y$ in $\R^{\mc I}$ have entries, respectively,  
$$(x\vee y)_i=\max\{x_i, y_i\},\quad (x\wedge y)_i=\min\{x_i, y_i\}\,,\quad \forall i\in\mc I\,.$$
We use the notation $\vee L$ and $\wedge L$ to indicate the entry-wise supremum and infimum, respectively, of a subset $L\subseteq \R^{\mc I}$. 
{ 
	A partially ordered set $(L, \geq)$ 
	is a \textit{lattice} if for all pairs $x, y \in L$ the least upper bound $x\vee y$ and the greatest lower bound $x\wedge y$ exist in $L$.  A lattice is \textit{complete} if in addition, for any subset $T$ of $L$, the elements $\vee T$ and $\wedge T$ lie in $L$. In this case,  $\vee L$ and $\wedge L$ are referred to as the \emph{greatest} and the \emph{least} elements of $L$, respectively.  
}


\begin{definition}
A binary action game is \emph{super-modular} if for every player $i$ in $\mc V$,
\be\label{increasing-differences}
u_i(1,x_{-i})-u_i(-1,x_{-i}) \geq u_i(1,y_{-i})-u_i(-1,y_{-i})\,,\ee
whenever $x_{-i}\ge y_{-i}$. 
\end{definition}

In other words, a game is super-modular if the marginal utility $u_i(1,x_{-i})-u_i(-1,x_{-i})$  of every player $i$ is a monotone nondecreasing function of the strategy profile  $x_{-i}$ of the other players. This is sometimes referred to as the \emph{increasing difference property}. 

Let us introduce the four maps \be\label{4maps}f^+, f^-,g^{+},g^{-}:\mc X\to\mc X\ee on the configuration space, respectively defined by 
$$\begin{array}{ll}f^+(x)=\,\bigvee\{y\in\mc X\,|\, x\uparrow y\}\,,\ &
f^-(x) =\,\bigwedge\{y\in\mc X\,|\,x\downarrow y\}\,,\\[5pt]
g^{+}(x)=\,\bigvee\{y\in\mc X\,|\, x\upharpoonleft y\}\,,\ &
g^{-}(x)=\,\bigwedge\{y\in\mc X\,|\, x\downharpoonright y\}\,,
\end{array}$$
for every $x$ in $\mc X$. 
Thanks to Lemma \ref{lemma:suprmodular-paths} in Appendix \ref{sec:proof-supermodular}, we have that $f^+(x)$ and $g^+(x)$ ($f^{-}(x)$ and $g^{-}(x)$) represent the maximal (minimal) configurations that are BR-reachable and, respectively, I-reachable from $x$ by a monotone (anti-monotone) path. Notice that all sets in the righthand side of the above are nonempty as they contain the strategy profile $x$. 

In the following statement, we gather a number of properties relating equilibria of super-modular games with the behavior of the maps \eqref{4maps} that will be used in the rest of the paper. 
\begin{proposition}\label{prop:stability-supermodular} 
Consider a finite super-modular game with binary action sets. Then, for every configuration $x$ in $\mc X$,
\begin{enumerate}
\item[(i)]  $f^-(f^{+}(x))\in\mc X^*$ and $f^+(f^{-}(x))\in\mc X^*$ are, respectively, the greatest and least equilibria that are I-reachable from $x$;
\item[(ii)] $f^-(g^{+}(x))\in\mc X^*$ and $f^+(g^{-}(x))\in\mc X^*$ are, respectively, the greatest and least  equilibria BR-reachable from $x$; 
\item[(iii)]if $x\in\mc X^*$, then $g^{+}(x),g^{-}(x)\in\mc X^*$ are, respectively, the greatest and least equilibria BR-reachable from $x$. 
\end{enumerate}
Moreover,
\begin{enumerate}
\item[(iv)] the set of equilibria $\mc X^*$ is a nonempty complete lattice and it is globally I-stable;
\item[(v)] the two configurations \be\label{extremalNash}\begin{array}{l} \ul x^*=f^+(g^{-}(-\1))=f^+(-\1)\,,\\[5pt]\ov x^*=f^-(g^{+}(+\1))=f^-(+\1)\,,\end{array}\ee
are, respectively, the least and greatest equilibria;
\item[(vi)] for two equilibria $x$ and $y$ in $\mc X^*$, $f^{+}(x\vee y)$ and $f^{-}(x\wedge y)$ are, respectively, the least and the greatest equilibria that are, respectively, above and below both $x$ and $y$;
\item[(vii)] if $\ul x^*= \ov x^*=x^*$, then  $\mc X^*=\{x^*\}$ is globally BR-stable.
\end{enumerate}
\end{proposition}
\begin{proof}
See Appendix \ref{sec:proof-supermodular}. 
\end{proof}
\begin{remark}
In general the set of equilibria $\mc X^*$ is not BR-invariant for a super-modular game: e.g., a 2-player game with utilities $u_1(x_1,x_2)=x_1x_2$ and $u_2(x_1,x_2)=0$ for $x_1$ and $x_2$ in $\{\pm1\}$ is super-modular but its equilibrium set $\mc X^*=\{\pm\1\}$ is not BR-invariant since $\mc B_2(+1)=\mc B_2(-1)=\{\pm1\}$. 
\end{remark}

\subsection{Analysis of network coordination games}
We now introduce binary network coordination games on general directed weighted networks. 
\begin{definition}\label{def:coordination} For a network $\mc G=(\mc V,\mc E,W)$ and a vector $h$ in $\R^{\mc V}$, the \emph{coordination game} on $\mc G$  with \emph{external field} $h$ is the game with player set $\mc V$, whereby every player $i$ in $\mc V$ has  binary action set $\mc A=\{\pm1\}$ and utility function 
\be\label{utility:coordination} u_i(x)=x_i\sum_{j\in\mc V}W_{ij}x_j+h_ix_i\,.\ee 
\end{definition}


To emphasize the dependance on the external field $h$, we shall use the notation $\nash$ for the set of equilibria of the network coordination game with external field $h$, and write $x\stackrel{h}{\to}y$ if a configuration $y$ is I-reachable from $x$ in this game.   Observe that the utility function \eqref{utility:coordination} may be rewritten as 
\be\label{eq:utility2}u_i(x)=x_i\left(h_i+w_i^+(x)-w_i^-(x)\right)\,,\ee
where $w_i^+(x)$ and $w_i^-(x)$ are defined as in \eqref{wi+wi-}. 
Equation \eqref{eq:utility2} highlights the decomposition of the utility of a player $i$ in $\mc V$ as the sum of a term $h_ix_i$ that depends only on her own action and rewards its alignment  with the corresponding entry of the external field, plus a term $x_i(w_i^+(x)-w_i^-(x))$ that is the difference between the aggregate weight of links pointing to out-neighbors with the same action and the aggregate weight of links pointing to out-neighbors with the opposite action.

The following statement gathers a few simple results on coordination games.
\begin{lemma}\label{lemma:coordination}
Consider the coordination game on a network $\mc G=(\mc V,\mc E,W)$ with external field $h$ in $\R^{\mc V}$. Then, 
\begin{enumerate}
\item[(i)] for every strategy profile $x$ in $\mc X$ and every player $i$ in $\mc V$,
\be\label{eq:Nash2}x_i\in\mc B_i(x_{-i})\qquad\Longleftrightarrow\qquad u_i(x)\ge0\,;\ee
\item[(ii)]  the best response correspondence has the threshold form 
\be\label{eq:BR}\mc B_i(x_{-i})=\left\{\ba{lcl}\{+1\}&\se&w_i^+(x)>r_iw_i\\[3pt]\{\pm1\}&\se&w_i^+(x)=r_iw_i\\[3pt]\{-1\}&\se&w_i^+(x)<r_iw_i\,,\ea\right.\ee
for every $i$ in $\mc V$, where 
$r_i=\frac12-\frac{h_i}{2w_i}$
is the threshold of player $i$; 
\item[(iii)] action $a=\pm1$ is a strictly dominant strategy for a player $i$ in $\mc V$ if and only if 
$ah_i>w_i$; 
\item[(iv)] the game is super-modular; 
\item[(v)]  $\nash$ is a nonempty complete lattice and it is globally I-stable. 
\end{enumerate}
\end{lemma}
\begin{proof}
See Appendix \ref{sec:proof-lemma-coordination}. 
%
%
%
\end{proof}

The next result explicitly connects the notions of equilibrium and I-reachability for coordination games to those of absorbing configurations and finite time reachability for the LTD with constant external field.

\begin{lemma}\label{lemma:coordination=LTM}
Let $X(t)$ be the LTD on a network $\mc G=(\mc V,\mc E,W)$ with constant external field $h$ in $\R^{\mc V}$ and initial configuration $X(0)= x^{(0)}$ in $\mc X$. Then, for every $x^*$ in $\mc X$: 
\begin{enumerate}
\item[(i)]  $ x^{(0)}\stackrel{h}{\to} x^*$ if and only if with probability $1$ there exists $t\ge0$ such that $X(t)=x^*$; 
\item[(ii)]  $ x^{(0)}\in\nash$ if and only if $ x^{(0)}$ is an absorbing configuration for $X(t)$, i.e., with probability $1$, $X(t)= x^{(0)}$ for every $t\ge0$. 
\end{enumerate}
\end{lemma}
\begin{proof} See Appendix \ref{sec:proof-lemma-coordination=LTM}. 
\end{proof}

Lemmas \ref{lemma:coordination} and \ref{lemma:coordination=LTM} imply the following result, ensuring that, for every initial configuration $X(0)$, with probability $1$ the LTD $X(t)$ on a network $\mc G$ with constant external field $h$ always get absorbed if finite time in the set of equilibria $\mc N_h$ of the coordination game on $\mc G$ with external field $h$. 

\begin{theorem}\label{theo:coordination=LTM}
Let $X(t)$ be the LTD on a network $\mc G=(\mc V,\mc E,W)$ with constant external field $h$ in $\R^{\mc V}$. Then, 
with probability $1$, there exists $t_0\ge0$ such that 
$$X(t)\in\nash\,,\qquad \forall t\ge t_0\,.$$  
\end{theorem}
\begin{proof} The claim follows directly from Lemma \ref{lemma:coordination} (v) and Lemma \ref{lemma:coordination=LTM}. \end{proof}

Theorem \ref{theo:coordination=LTM} implies that 
convergence with probability $1$ to consensus configurations is guaranteed for LTD  on a network $\mc G$ with constant external field $h$  is equivalent to the non-existence of coexistent equilibria of the corresponding coordination game. To study this, it is first convenient to introduce the notation
$$\nashconsensus=\nash\cap\{\pm\1\}\,,\qquad\nashcoexistent=\nash\setminus\{\pm\1\}\,,$$ for the subsets of consensus and, respectively, co-existent equilibria of the coordination game on $\mc G$ with external field $h$.  
We then introduce the following notion. 

\begin{definition}\label{def1} A coordination game on a network $\mc G$ with external field $h$ is
\begin{itemize}
\item \emph{regular} if $|\nashconsensus|=2$;
\item \emph{biased} if $|\nashconsensus|=1$; more precisely, for $a=\pm1$, the network coordination game is $a$-biased if $\nashconsensus=\{a\}$;
\item \emph{frustrated} if $|\nashconsensus|=0$. 
\end{itemize}
\end{definition}
Notice that, in a frustrated network coordination game, neither of the consensus configurations is an equilibrium. Since $\nash$ is never empty, 
a frustrated network coordination game always admits at least one co-existent equilibrium. In contrast, when the game is not frustrated (either regular or biased) at least one consensus configuration is an equilibrium. Besides consensus,  there might or might not exist co-existent equilibria. 
To distinguish these cases, the following further classification proves useful.
\begin{definition}\label{def2}  A coordination game on a network $\mc G$ with external field $h$ is
\begin{itemize}
\item \emph{unpolarizable} if $ \nashcoexistent=\emptyset$;
\item \emph{polarizable} if $ \nashcoexistent\neq\emptyset$.
\end{itemize}
\end{definition}
The set of equilibria of an unpolarizable regular network coordination game contains both consensus configurations $\pm\1$ and no other configurations, whereas unpolarizable biased network coordination games admit a single (consensus) equilibrium: $x^*=+\1$ in the positively biased case and $x^*=-\1$ in the negatively biased one. On the other hand, polarizable network coordination games always admit co-existent equilibria possibly in addition to consensus ones (if they are regular or biased). In the sequel, we shall identify necessary and sufficient conditions for a network coordination game to be regular, biased, or frustrated, and for it to be polarizable or unpolarizable.

%
%
%
%
%
%
%
We now introduce two sets that will play a key role in our analysis:
\be\label{stubborn}\mc S_a(h)=\{i\in \mc V\,|\, ah_i>w_i\}\,,\qquad a=\pm1\,.\ee
By Lemma \ref{lemma:coordination} (iii), $\mc S_a(h)$ coincides with the set of players for which $a$ is a strictly dominant action, i.e., the set of $a$-stubborn agents. The following simple result relates the presence of $a$-stubborn players with that of the consensus equilibrium $-a\1$. 

\begin{lemma}\label{lemma:consensus}
Consider a coordination game on a network $\mc G=(\mc V,\mc E,W)$ with external field $h$ and let $\ul x^*(h)$ and $\ov x^*(h)$ be its least and greatest equilibria, respectively. 
Then, 
\begin{enumerate}
\item[(i)] $\ul x^*(h)=-\1\;\Leftrightarrow\; \mc S_{+1}(h)=\emptyset\;\Leftrightarrow\;h\le w$ 
\item[(ii)] $\ov x^*(h)=+\1\;\Leftrightarrow\; \mc S_{-1}(h)=\emptyset\;\Leftrightarrow\;h\ge -w$ 
\end{enumerate}
\end{lemma}
\begin{proof} See Appendix \ref{sec:proof-lemma-consensus}
\end{proof}

In fact, Lemma \ref{lemma:consensus} directly implies the following result.

\begin{proposition} \label{prop:regular-biased-frustrated} 
Let $\mc G$ be a network with out-degree vector $w$. 
Then, the coordination game on $\mc G$ with external field $h$ is:
\begin{enumerate}
\item[(i)] regular if and only if  
\be\label{Nostubborn}-w\le h\le w\,;\ee
\item[(ii)] $a$-biased for $a=\pm1$ if and only if 
\be\label{1stubborn}  w\ge-ah\,,\qquad w\ngeq ah
 \,;\ee 
\item[(iii)] \emph{frustrated} if and only if 
\be\label{2stubborn}w\ngeq-h\,,\qquad w\ngeq h\,.\ee
\end{enumerate}
\end{proposition}

\begin{remark}\label{remark:stubborn}
The necessary and sufficient conditions in Proposition \ref{prop:regular-biased-frustrated} can be readily interpreted in terms of the presence of stubborn agents, as introduced in Section \ref{sec:game-theoretic}. 
In fact, Lemma \ref{lemma:consensus} implies that \eqref{Nostubborn} is equivalent to the fact that no player is stubborn, \eqref{1stubborn} is equivalent to the existence of at least one $a$-stubborn agent but no $-a$-stubborn agents, and \eqref{2stubborn} is equivalent to the existence of both $+1$- and $-1$-stubborn agents. Hence, Proposition \ref{prop:regular-biased-frustrated} states that a coordination game is regular if and only if there are no stubborn agents, biased if and only if it contains stubborn agents of one type only, and frustrated if it contains stubborn agents of both types. 
\end{remark}

In contrast to the relative simplicity of the characterization above, necessary and sufficient conditions for polarizability of  network coordination games as per Definition \ref{def2}  are in general more involved and rely on the notion of indecomposability introduced in Definition \ref{def:decomposition}.  



\begin{proposition}\label{prop:polarization} The coordination game on a network $\mc G$ with external field $h$ is unpolarizable if and only if  $\mc G$ is $h$-indecomposable.
\end{proposition}
\begin{proof} See Appendix \ref{sec:proof-prop-polarization}.
\end{proof}

\begin{remark} \label{remark:cohesive}
Given a graph $\mc G=(\mc V,\mc E,W)$ and $r$ in $[0,1]$, a subset of nodes $\mc S\subseteq\mc V$ is called $r$-\emph{cohesive} \cite{Morris:2000}  if 
\be\label{eq:cohesiveness}w_i^{\mc S}\ge rw_i\,,\qquad\forall i\in\mc S\,,\ee
and $r$-\emph{closed} if its complement $\mc V\setminus\mc S$ is $(1-r)$-cohesive.
Notice that, using the identity $w_i=w_i^{s}+w_i^{-s}$, condition \eqref{eq:robustly_indec} in the special case $h^-=h^+=h$ can be rewritten as 
$2w_i^{s} < w_i-sh_i\,.$
In the special case when players have homogeneous thresholds $r_i=r$ in $[0,1]$, equivalently when the external field is proportional to the node degree vector, i.e., 
$h=(1-2r)w$, 
\eqref{eq:robustly_indec} is equivalent to  $w_i^{s}<rw_i$. 
Hence, in this special case, $h$-indecomposability of a graph $\mc G$ is equivalent to the non-existence of nonempty proper subsets of nodes $\mc S$ that are both $r$-cohesive and $r$-closed. In this sense, Proposition \ref{prop:polarization} generalizes \cite[Proposition 9.7]{Jackson:2008} to coordination games on weighted directed networks with heterogeneous thresholds.
\end{remark}

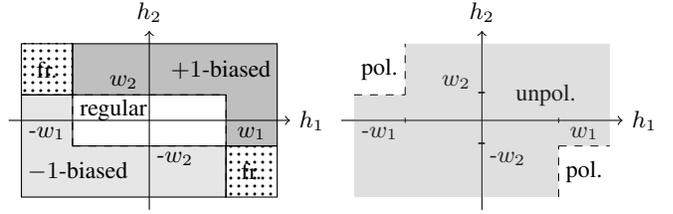
\begin{figure}
	\centering
	\begin{tikzpicture}[scale=0.34]
		\draw[fill=white!50!gray](-3,1) -- (3,1) -- (3,-1) -- (5,-1) -- (5,3)-- (-3,3) --cycle;
		\draw[fill=white!80!gray] (-3,-1) -- (3,-1) -- (3,-3) -- (-5,-3) -- (-5,1)-- (-3,1) --cycle; 
		
		\draw[->] (-5.5, 0) -- (5.5, 0) node[right] {\small$h_1$};
		\draw[->] (0, -3.5) -- (0, 3.5) node[above] {\small$h_2$};
		\draw[scale=1, domain=-5:3, dashed, variable=\x] plot ( {\x},{1});
		\draw[scale=1, domain=-3:5, dashed, variable=\x] plot ( {\x},{-1});
		\draw[scale=1, domain=-3:1, dashed, variable=\y] plot ( {3},{\y});
		\draw[scale=1, domain=-1:3, dashed, variable=\y] plot ({-3}, {\y});	
		\draw[pattern=dots] (3,-1) -- (5,-1) -- (5,-3) -- (3,-3) --cycle;
		\draw[pattern=dots] (-3,1) -- (-3,3) -- (-5,3) -- (-5,1) --cycle;
		
		
		
		\node at (4,-0.5) {\small$w_1$};
		\node at (-4,-0.5) {\small-$w_1$};
		\node at (-1,1.5) {\small$w_2$};
		\node at (1,-1.5) {\small-$w_2$};
		\node at (-4,2) {\small fr.};
		\node at (4,-2) {\small fr.};
		\node at (2.8,2) {\small$+1$-biased};
		\node at (-2.8,-2) {\small$-1$-biased};
		\node at (-1.4,0.4) {\small regular};
	\end{tikzpicture}\hspace{0.1cm}\begin{tikzpicture}[scale=0.34]
		\draw[->] (-5.5, 0) -- (5.5, 0) node[right] {\small$h_1$};
		\draw[->] (0, -3.5) -- (0, 3.5) node[above] {\small$h_2$}; 
		\draw[scale=1, domain=-5:-3, dashed, variable=\x] plot ( {\x},{1});
		\draw[scale=1, domain=3:5, dashed, variable=\x] plot ( {\x},{-1});
		\draw[scale=1, domain=-3:-1, dashed, variable=\y] plot ( {3},{\y});
		\draw[scale=1, domain=1:3, dashed, variable=\y] plot ({-3}, {\y});	
		\node at (4,-0.5) {\small$w_1$};
		\node at (-4,-0.5) {\small-$w_1$};
		\node at (-1,1.5) {\small$w_2$};
		\node at (1,-1.5) {\small-$w_2$};
		\node at (0,1) {-};
		\node at (0,-1) {-};
		\draw[scale=1, domain=-0.1:0.1, variable=\y] plot ({-3}, {\y});	
		\draw[scale=1, domain=-0.1:0.1, variable=\y] plot ({3}, {\y});	
		\node at (-4,2) {\small pol.};
		\node at (4,-2) {\small pol.};
		\node at (2.5,1) {\small unpol.};
		\fill[gray, nearly transparent] (3,-1) -- (5,-1) -- (5,3)-- (-3,3) --(-3,1)--(-5,1)--(-5,-3)-- (3,-3)--cycle;
	\end{tikzpicture}
	
	\caption{Classification  of two-player network coordination games as in Example \ref{ex:regular-biased-frustrated}, based on Proposition \ref{prop:regular-biased-frustrated} and \ref{prop:polarization}. 
	}
	\label{fig:regular-biased-frustrated}
\end{figure}

\begin{example}\label{ex:regular-biased-frustrated}
		Let $\mc G=(\mc V,\mc E,W)$ be a graph with two nodes $\mc V = \{1,2\}$ connected by two directed links of weight $W_{12}=w_1$ and $W_{21}=w_2$, respectively. 
		Consider a network coordination game on $\mc G$ with external field $h=(h_1,h_2)$. 
		
		As illustrated in Figure \ref{fig:regular-biased-frustrated}, by Proposition \ref{prop:regular-biased-frustrated}, the network coordination game is: 
		regular if $\abs{h}\leq w$ (white region); 
		$+1$-biased if $h\geq-w$ and $h\nleq w$ (dark gray region);
		$-1$-biased if $h\leq w$ and $h\ngeq-w$ (light gray region);
                 frustrated  if $h_1>w_1$ and $h_2<-w_2$ or $h_1<-w_1$ and $h_2>w_2$ (dotted region). 
                 
		On the other hand, Proposition \ref{prop:polarization} 
		ensures that the network coordination game on $\mc G$ is  unpolarizable if and only if one of the following holds true: (i) $h<w$, (ii) $h>-w$, (iii) $\abs{h_1}<w_1$, (iv) $\abs{h_2}<w_2$, as illustrated in Figure \ref{fig:regular-biased-frustrated}. Notice that, for the special case of only two-players, the network coordination game is polarizable if and only if it is frustrated.
\end{example}

Biased unpolarizable coordination games can be characterized in an equivalent simpler form.

\begin{proposition}\label{prop:biased-indecomposable} The coordination game on a network $\mc G$ with external field $h$ has unique equilibrium $x^*=a\1$ for a in $\{\pm1\}$ if and only if the following conditions are both satisfied: 
\begin{enumerate}
\item[(a)] $w\ngeq ah$; 
\item[(b)] every non-empty subset $\mc R\subseteq\mc V\setminus\mc S_a(h)$ contains some node $i$ such that 
\be\label{indecomposable-simple}w^{\mc R}_i<w^{\mc V\setminus\mc R}_i+ah_i\,.\ee
%
%
%
%
%
%
%
%
%
%
\end{enumerate}
\end{proposition}
\begin{proof} See Appendix \ref{sec:proof-prop-biased-indecomposable}. 
\end{proof}

\begin{remark}
In the special case $h=(1-2r)w$ considered in Remark \ref{remark:cohesive}, i.e., when players have homogeneous thresholds $r_i=r$ in $[0,1]$, \eqref{indecomposable-simple} is equivalent to  $w_i^{\mc R}<(1-r)w_i$, so 
condition (b) of Proposition \ref{prop:biased-indecomposable} reduces to the non-existence of $(1-r)$-cohesive subsets of $\mc V\setminus\mc S_a(h)$. In the literature \cite{Morris:2000}, such property is referred to as the set $\mc V\setminus\mc S_a(h)$ being uniformly not $(1-r)$-cohesive. In this sense, Proposition \ref{prop:biased-indecomposable} generalizes \cite[Proposition 9.8]{Jackson:2008} to network coordination games with heterogeneous thresholds.
\end{remark}
We conclude this section with the statement below, gathering some results on global I- and BR-stability of consensus equilibria for network coordination games that directly follow from the analysis just developed.
\begin{corollary}\label{coro:stability} For a graph $\mc G$ with out-degree vector $w$,   
consider the network coordination game on it with external field $h$. Assume that $\mc G$ is $h$-indecomposable. Then: 
\begin{enumerate}
\item[(i)] if $|h|\leq w$, then $\nash=\{\pm\1\}$ is globally I-stable; 
\item[(ii)] if $|h|< w$, then $\nash=\{\pm\1\}$ is globally BR-stable; 
\item[(iii)] if 
$w\ngeq ah$, 
then $\nash=\{a\1\}$ is globally BR-stable.
\end{enumerate}
\end{corollary}
\begin{proof}
(i) Proposition \ref{prop:regular-biased-frustrated} (i) implies that when $|h|\leq w$, the game is regular so that the two consensus configurations $-\1$ and $+\1$ are both equilibria. Since $\mc G$ is $h$-indecomposable, Proposition \ref{prop:polarization} guarantees that the game is unpolarizable. Then, the set of equilibria is  $\nash=\{\pm\1\}$ and Proposition \ref{prop:stability-supermodular} (iv) guarantees that it is globally I-stable.

(ii) We already know from point (i) that $\nash=\{\pm\1\}$ is globally I-stable. The BR-stability follows from the fact that when $|h|< w$ these two equilibria are strict.

(iii) By Lemma \ref{lemma:consensus}, $\mc S_a(h)\neq \emptyset$, so that the configuration $-a\1$ is not an equilibrium. 
Since $\mc G$ is $h$-indecomposable, Proposition \ref{prop:polarization} implies that $ \nashcoexistent=\emptyset$. 
Hence,  $\nash=\{a\1\}$. BR-stability then follows from Proposition \ref{prop:stability-supermodular} (vii).
\end{proof}

\section{Robust stability}

In this section, we first introduce and characterize robust versions of the notions introduced in Definitions \ref{def1} and \ref{def2}.  This will pave the way to the proof of Theorem \ref{theo:time-varying} on the asymptotic behavior of the ATV-LTD. 

\subsection{Robustness network coordination}
\label{sec:robust-coordination}
For a graph $\mc G= (\mc V,\mc E,W)$  and  a subset of vectors $\mc H\subseteq \R^{\mc V}$, we  say that a property of the network coordination game on $\mc G$ is
satisfied $\mc H$-\emph{robustly} if it is satisfied \emph{for every}  external field $h$ in $\mc H$.
In what follows we concentrate on the special case when, for two vectors $h^-$ and $h^+$ in $\R^{\mc V}$ such that $h^-\le h^+$, the set $\mc H$ is the hyper-rectangle \eqref{H}. 
In this case, the verification that certain important properties are $\mc H$-robustly satisfied can be significantly simplified with respect to checking the property for every single value of $h$ in $\mc H$. Below we report the results in this sense. We start with the following robust version of 
Proposition \ref{prop:regular-biased-frustrated}.

\begin{corollary} \label{coro:robustly-regular+biased} 
Let $\mc G$ be a graph with out-degree vector $w$ and let $\mc H$ be as in \eqref{H} for two vectors $h^-\le h^+$.
Then, the network coordination game on $\mc G$  is:
\begin{enumerate}
\item[(i)] $\mc H$-robustly regular if and only if  
\be\label{w>=max}w\ge h^+\,,\qquad w\ge-h^-\,;\ee
\item[(ii)] $\mc H$-robustly $a$-biased for an action $a=\pm1$ if and only if 
\be\label{1stubborn-bis} 
 w\ge-ah^{-a}\,,\qquad w\ngeq a h^{-a}
 \,;\ee 
\item[(iii)] $\mc H$-robustly frustrated if and only if 
$$
w\ngeq-h^{+}\,,\qquad w\ngeq h^{-}\,.$$
\end{enumerate}
\end{corollary}

Another interesting property is the robust unpolarizability, which can be interpreted as the resilience of a network coordination game against getting co-existent equilibria. By virtue of Proposition \ref{prop:polarization}, this can equivalently be expressed as a robust indecomposability of the graph $\mc G$. However, it is useful to reformulate this in a form analogous to Definition \ref{def:decomposition}, as in the following result.

\begin{theorem}\label{theorem:robust-indecomposability}Let $\mc G=(\mc V,\mc E,W)$ be a graph and let $\mc H$ be as in \eqref{H} for two vectors $h^-\le h^+$. Then, the following conditions are equivalent:
\begin{enumerate}
\item[(a)] the network coordination game on $\mc G$ is $\mc H$-robustly unpolarizable;
\item[(b)]  $\mc G$ is $\mc H$-robustly indecomposable;
\item[(c)] $\mc G$ is $(h^-,h^+)$-indecomposable. 
\end{enumerate}
\end{theorem}
\begin{proof}
Clearly, equivalence between conditions (a) and (b) directly follows from Proposition \ref{prop:polarization}. We shall now prove equivalence between conditions (a) and (c). 

First, assume that the network coordination game on $\mc G$ is not $\mc H$-robustly unpolarizable, i.e., there exists an external field $h$ in $\mc H$ such that the set of equilibria $\nash$ contains a co-existent configuration $x^*\ne\pm\1$. 
Notice that $h^-_i\le h_i\le h^+_i$ implies that $x_i^* h_i\le s_ih^{s_i}$, where $s_i=\sgn(x_i^*)$, for every player $i$ in $\mc V$. 
Then,  Lemma \ref{lemma:coordination} (i) and \eqref{eq:utility2} imply that
$$\ba{rcl}
0&\le& u_i(x^*)\\[5pt]
&=&x_i^* h_i+x_i^*w_i^+(x^*)-x_i^*w_i^-(x^*)\\[5pt]
&\le&s_ih^{s_i}+w_i^{{s_i}}-w_i^{-s_i}\,.
\ea$$
Hence, there exists no node $i$ in $\mc V$ satisfying \eqref{eq:robustly_indec} for the nontrivial binary partition $\mc V=\mc V_{x^*}^+\cup\mc V_{x^*}^-$.  This proves that condition (c) is not satisfied.

On the other hand, if condition (c) is not satisfied, then Proposition \ref{prop:indecomposable-necessary} and Lemma  \ref{lemma:coordination=LTM} (ii) imply that there exist $h^*$ in $\mc H$ and $x^*$ in $\mc X\circ_{h^*}$, so  that the network coordination game on $\mc G$ with external field $h^*$ is polarizable, hence condition (a) is not satisfied.
\end{proof}

We now focus on the stability results in Corollary \ref{coro:stability}, for which, besides their straightforward robust generalization, some deeper consequences can be derived as reported below. These will in turn prove instrumental for the analysis carried on in next section.

\begin{theorem}\label{theo:robustreachable} 	
	Let $\mc G=(\mc V,\mc E,W)$ be a graph of order $n$ and let $\mc H$ be as in \eqref{H} for two vectors $h^-\le h^+$. Assume that $\mc G$ is $\mc H$-robustly indecomposable. Then,
	\begin{enumerate}
		\item[(i)] for every configuration $x$ in $\mc X$ there exists an $\mc H$-robust $I$-path from $x$ to $\{\pm\1\}$ of length at most $n$. In particular, $\{\pm\1\}$ is $\mc H$-robustly globally I-reachable; 
		\item[(ii)]  if the network coordination game on $\mc G$  is $\mc H$-robustly regular, then the equilibrium set $\mc X^*=\{\pm\1\}$ is $\mc H$-robustly globally I-stable;  
		\item[(iii)]  if there exists an action $a=\pm1$ such that \eqref{1stubborn-bis} holds true, then there exists an $\mc H$-robust $I$-path from every configuration $x$ in $\mc X$ to $a\1$. In particular, in this case, the equilibrium set $\mc X^*=\{a\1\}$ is $\mc H$-robustly globally BR-stable.
	\end{enumerate} 
	
\end{theorem}
\begin{proof}
(i)		Given an arbitrary configuration $x$ in $\mc X$, consider the external field $h^x$ in $\mc H$ with entries
		$$h^x_i=h^{x_i}_i\,,\qquad i\in\mc V\,,$$
and let 
		$$\ul x=f^+(f^{-}(x, h^x),h^x)\,,\qquad \ov x=f^-(f^{+}(x, h^x), h^x)\,.$$ 
		By Proposition \ref{prop:stability-supermodular} (i), the above are  
		the least and greatest equilibria of the network coordination game on $\mc G$ with external field $h^x$ that 
		are I-reachable from configuration $x$. 
		As the graph $\mc G$ is $\mc H$-robustly indecomposable, it follows from Theorem \ref{theorem:robust-indecomposability} that the network coordination game on $\mc G$ is $\mc H$-robustly unpolarizable. Since $h^x\in\mc H$, this implies  that the network coordination game on $\mc G$ with external field $h^x$ is unpolarizable, so that its equilibria $\ul x$ and $\ov x$  are both consensus configurations.  
		Since clearly $\ul x\le\ov x$, there are three possible alternative cases: 
		\begin{enumerate}
		\item[(a)] $\ul x=\ov x= +\1$; 
		 \item[(b)] $-\1= \ul x <\ov x=+\1$; 
		 \item[(c)]$\ul x=\ov x=-\mathbf{1}$. 
		 \end{enumerate}
In both cases (a) and (b) we have 
$$+\1=\ov x=f^{-}(f^{+}(x, h^x), h^x)=f^{+}(x, h^x)\,,$$ so that $x\uparrow+\1$ for the network coordination game on $\mc G$ with external field $h^x$.  
		We now show that in fact a monotone I-path from $x$ to $+\1$ exists $\mc H$-robustly by proving that $f^{+}(x, h^-)=f^{+}(x, h^x)$. Indeed, since in any monotone path only players originally playing action $-1$ can get activated, and since $h^-_i=h^x_i$ for every such player, we have that every monotone path is an $I$-path with respect to $h^-$ if and only if it is an $I$-path with respect to $h^x$. By the way the function $f^+$ is defined, this yields that 
		$$f^{+}(x, h^-)=f^{+}(x, h^x)=f^-(f^{+}(x, h^x), h^x) =+\1\,,$$ which shows that $+\mathbf{1}$ is reachable from $x$ by a monotone $I$-path for the network coordination game on $\mc G$ with external field $h^-$. We can then consider any monotone path $(x^{(0)},x^{(1)},\ldots x^{(l)})$ from $x^{(0)}=x$ to $x^{(l)}=f^{+}(x, h^-)=\mathbf{1}$ that is an $I$-path for the network coordination game on $\mc G$ with external field $h^-$, where clearly $l \leq |\mc V|$. A direct monotonicity argument shows that this is also a monotone $I$-path for the network coordination game on $\mc G$ with any external field $h$ in $\mc H$. 
		
	Similarly, in both cases (b) and (c) we have 	
	$$- \mathbf{1}=\ul x=f^-(f^{+}(x, h^x), h^x)=f^-(x, h^x)\,,$$ 
so that by an argument completely analogous to the one developed above we can find an $\mc H$-robust anti-monotone $I$-path from $x$ to $-\1$. The proof of point (i) is then completed by the observation that the length of monotone and anti-monotone is never larger than $n$. 
		
		(ii) From point (i), the equilibrium set $\mc X^*=\{\pm\1\}$ is $\mc H$-robustly globally I-reachable. Due to the $\mc H$-robust regularity assumption, $\mc X^*=\{\pm\1\}$ is also $\mc H$-robustly I-invariant. Hence, $\mc X^*=\{\pm\1\}$ is $\mc H$-robustly globally I-stable. 
		
		(iii)  Existence of an $\mc H$-robust $I$-path from every configuration $x$ in $\mc X$ to the equilibrium $a\1$ follows from point (i) and Corollary \ref{coro:robustly-regular+biased} (ii). On the other hand, Corollary \ref{coro:stability} (iii) ensures BR-invariance of $a\1$ for the network coordination game on $\mc G$ with any external field $h$ in $\mc H$. Therefore, the equilibrium set $\mc X^*=\{a\1\}$ is $\mc H$-robustly globally BR-stable.
\end{proof}

\subsection{Proof of Theorem \ref{theo:time-varying}}\label{sec:coordination-dynamics}
We now apply the results of Section \ref{sec:robust-coordination} to prove Theorem \ref{theo:time-varying}. Recall that the ATV-LTD $X(t)$ on a network $\mc G=(\mc V,\mc E,W)$ with external field $h(t)$ is a continuous-time inhomegeneous Markov chain $X(t)$, whereby agents $i$ in $\mc V$ get activated at the ticking of independent rate-$1$ Poisson clocks and, when activated at time $t\ge0$, they modify their state according to the update rule \eqref{dynamics}. 

We shall denote by $\Lambda(t)$ in $\R^{\mc X\times\mc X}$ the transition rate matrix of the continuous-time Markov chain $X(t)$, whose entries $\Lambda_{xy}(t)$ stand for the transition rates from configuration $x$ in $\mc X$ to configuration $y$ in $\mc X$ at time $t\ge0$. Notice that $\Lambda(t)$ depends on the external field $h(t)$ and for this reason it is time-varying.  We have that $\Lambda_{xy}(t)=0$ whenever $x$ and $y$ differ in more than one entry, reflecting the fact that with probability $1$ no two agents will modify their action simultaneously. On the other hand, if there exists $i$ in $\mc V$ such that $x_{-i}=y_{-i}$ and $x_{i}\ne y_i$, then 
$$\Lambda_{xy}(t)=\left\{\ba{lcl}1&\se&y_i(\sum_jW_{ij}x_i+h_i(t))>0\\[3pt]
0&\se&y_i(\sum_jW_{ij}x_i+h_i(t))\le0\,.\ea\right.$$ 
Finally, the diagonal entries of $\Lambda(t)$ are nonpositive and such that every row sum is zero, i.e., 
$\Lambda_{xx}(t)=-\sum_{y\ne x}\Lambda_{xy}(t)$.

Notice that the form of the update rule \eqref{dynamics} of the ATV-LTD implies that the following uniform bounds hold true for the transition rates of $X(t)$ at any time $t\ge0$:  
\be\label{estim} \Lambda_{xy}(t)>0\;\Rightarrow\; \Lambda_{xy}(t)\geq 1\,,\qquad \forall x,y\in\mc X\,,\ee 
\be\label{estim-2}\sum_{y\ne x}\Lambda_{xy}(t)\leq n\,,\qquad \forall x\in\mc X\,,\ee
where we recall that $n=\abs{\mc V}$ is the number of agents. 

(i) For every initial profile $X(0)=x^{(0)}$, Theorem \ref{theo:robustreachable} (i) guarantees the existence of an $\mc H$-robust $I$-path $(x^{(0)},x^{(1)},\ldots x^{(l)})$ 
of length  $l\leq n$ from $x^{(0)}$ to the set of consensus configurations $\{\pm\1\}$. 
Consider now the discrete-time jump chain \cite[p.~87]{norris1998markov} 
  associated to $X(t)$, defined by 
$$Y(k)=X(T_k)\,,\qquad k=0,1,\ldots$$
where $0=T_0< T_1< T_2<\ldots$ are the random times when the value of $X(t)$ changes. 
Using \eqref{estim} and \eqref{estim-2}, we can estimate the probability that the ATV-LTD follows this path at some time as follows  
$$\P\left(Y_{s+1}=x^{(1)},\dots, Y_{s+l}=x^{(l)}\Big|\, T_s=t, Y_s=x^{(0)}\right)\geq 1/n^{l}\,,$$
for every $s\ge0$. This implies that
$$\P(Y_{s+l}\in\{\pm\1\}\,|\, T_s=t,\,X_s=x_0)\geq1/n^l\ge 1/n^n\,,$$
for every $x_0$ in $\mc X$, $t\ge0$ and $s\ge0$.
A standard induction argument now yields that, for every initial condition $x^{(0)}$ in $\mc X$ and  for every $h=1,2,\dots$,
$$\P(Y_{s}\not\in\{\pm\1\}\,\forall s=0,\ldots,hn\,|\, X(0)=x^{(0)})\leq (1-1/n)^{hn}\,.$$
Since the bound above is uniform with respect to the initial configuration $x^{(0)}$ in $\mc X$, we get that 
$$\P(Y_{s}\not\in\{\pm\1\}\,\forall s=0,\ldots,hn)\leq (1-1/n)^{hn}\,.$$
Let now 
$T_{\pm\1}=\inf\{t\ge0:\,X(t)\in\{\pm\1\}\}$
be the (possibly infinite) first time that $X(t)$ is a consensus configuration.  
Then, 
$$\begin{array}{rcl}\P(T_{\pm\1}<+\infty)&=&1-\lim\limits_{t\to+\infty}\P(X(t)\ne\pm\1)\\
&=&1-\lim\limits_{h\to+\infty}\P(Y_{s}\ne\pm\1\,\forall s=0,\ldots,hn)\\
&\geq&1-\lim\limits_{h\to+\infty} (1-1/n)^{hn}\\
&=&1\,,\end{array}$$
thus proving that, with probability $1$, the set of consensus configurations $\{\pm\1\}$ is reached in finite time.

(ii)-(iii) We prove the two results at once. First, observe that the assumptions on $w$, $h^+$, and $h^-$ imply that in case (iii) $a\1$ is a strict equilibrium and in case (ii) $+\1$ and $-\1$ are both strict equilibria, for every $h^-\le h\le h^+$. Put $\mc X^*=\{a\1\}$ in the former case and $\mc X^*=\{\pm\1\}$ in the latter. Then, in both cases, $\Lambda_{xy}(t)=0$ for every $x\in\mc X^*$, $y\ne x$, and $t\ge0$, so that every  $x^*$ in $\mc X^*$ is an absorbing configuration for $X(t)$. In case (ii), \eqref{eq:caseii} then follows directly from point (i). 
In case (iii) instead, one can apply Theorem \ref{theo:robustreachable} (iii) and argue as in the proof of point (i) to show that 
the all-$a$ configuration $a\1$ is reached in finite time with probability $a$, i.e., $T_{a\1}=\inf\{t\ge0:\,X(t)=a\1\}$ satisfies 
$$\P(T_{a\1}<+\infty)=1\,.$$
Since configuration $a\1$ is absorbing in this case, \eqref{eq:caseii} follows. 
%

(iv) By assumption, there exist two players $i$ and $j$ in $\mc V$ such that $w_i< h_i^+$ and $w_j<-h_j^-$.  
Since $i\in\mc S_+(h^+)$ and $\mc G$ is $\mc H$-robustly indecomposable, Corollary \ref{coro:stability} implies that $\mc X^*_{h^+}=\{+\1\}$ is I-stable for the coordination game on $\mc G$ with external field $h^+$. This implies that, for every $\tau>0$, the ATV-LTD on $\mc G$ with  an external field $h(t)$  such that $h(t)=h^+$ for all $t$ in $[0,\tau)$, is such that  
$$\alpha_+=\min_{x\in\mc X}\P(X(\tau)=+\1|(X(0)=x))>0\,.$$ 
Analogously, since $j\in\mc S_-(h^-)$ and $\mc G$ is $\mc H$-robustly indecomposable, we get that the coordination dynamics with an external field that is constant $h(t)=h^-$ in the interval $[0,\tau)$ is such that  
$$\alpha_-=\min_{x\in\mc X}\P(X(\tau^-)=-\1|(X(0)=x))>0\,.$$ 
For the ATV-LTD with periodic piece-wise constant external field defined as follows
$$h(t)=\left\{\ba{lcll}h^+&\se& 2k\tau\le t<(2k+1)\tau&k\in\Z_+\\[5pt] h^-&\se&(2k+1)\tau\le t<(2k+2)\tau&k\in\Z_+\,,\ea\right.$$
we then have that 
$$\P(X((2k+1)\tau)=+\1,\,X((2k+2)\tau)=-\1|X(2k\tau)=x)\ge\alpha\,,$$
for every $k\in\Z_+$ and $x\in\mc X$, where 
$$\alpha=\alpha_+\alpha_->0\,.$$
It then follows that with probability $1$ there exist infinitely many nonnegative integer values of $k$ such that 
$$X((2k+1)\tau)=+\1\,,\qquad X((2k+2)\tau)=-\1\,$$ thus proving that $X(t)$ keeps fluctuating forever. 
\qed\medskip


\section{Conclusion}\label{sec:conclusion}
We have studied asynchronous time-varying linear threshold dynamics on general weighted directed networks of interacting agents, equipped with an external field modeling exogenous interventions or individual biases towards specific actions. We have proved necessary and sufficient conditions for global stability of consensus equilibria, robustly with respect to the (constant or time-varying) external field.

A key step in our analysis has consisted in the introduction of novel robust notions of improvement and best response paths. Our analysis has strongly relied on super-modularity of coordination games, but also their peculiar threshold structure of best response correspondences.  Extension of such concepts and results to more general super-modular games is a challenging problem that deserves further investigation. 

\appendices

\section{Proof of Proposition \ref{prop:indecomposable-necessary}}\label{sec:proof-prop-indecomposable-necessary}
That the network $\mc G$ is not $(h^-,h^+)$-indecomposable means that there exists a nontrivial binary partition of the node set as in \eqref{eq:binary-partition} such that 
$$w_i^{s}-w_i^{-s} + sh_i^{s}\ge 0\,,$$
for every $i$ in $\mc V_s$ and $s$ in $\{-,+\}$. 
Notice that the above can be rewritten as 
\be\label{non-dec}s(w_i^+-w_i^-+h_i^s)\,.\ee 
Now, let $h^*$ in $\R^{\mc V}$  be a vector with entries 
$$h^*_i=\left\{\ba{lcl}h^-_i&\se&i\in\mc V_-\\h^+_i&\se&i\in\mc V_+\,,\ea\right.$$
and let $x^*$ in $\mc X$ be a configuration with entries 
$$x^*_i=\left\{\ba{lcl}-1&\se&i\in\mc V_-\\+1&\se&i\in\mc V_+\,.\ea\right.$$
Clearly, $h^-\le h^*\le h^+$, so that $h$ belongs to $\mc H$. Moreover, the fact that $\mc V_-\ne\emptyset\ne\mc V_+$ and $\mc V_-\ne\mc V\ne\mc V_+$ implies that $x^*\ne\pm-\1$ is a coexistent configuration. 
Equation \eqref{non-dec} then implies that 
\be\label{fixed-point}x_i^*\left(\sum\nolimits_jW_{ij}x_j^*+h_i^*\right)=
x_i^*\left(w_i^+-w_i^-+h_i^*\right)
\ge0\,,\ee
for every $i$ in $\mc V$. Now, let $X(t)$ evolve according to the LTD on $\mc G$ with constant external field $h(t)=h^*$ and initial configuration $X(0)=x^*$.   
It then follows from \eqref{dynamics} and \eqref{fixed-point} that $X(t)=x^*$ for every $t\ge0$, thus proving the claim. \qed

\section{Properties of finite super-modular games}\label{sec:proof-supermodular}

In this Appendix, we prove the technical results of Section \ref{sec:supermodular} on super-modular games with binary actions. Throughout, we shall assume to have fixed a super-modular game with finite player set $\mc V$ and configuration space $\mc X=\{\pm1\}^{\mc V}$.

First, we state the following direct, though crucial, consequence of the increasing difference property \eqref{increasing-differences}.
\begin{lemma}\label{lemma:basic-supermod} For every player $i$ in $\mc V$, both $\mc B_i^+(x_{-i})$ and $\mc B_i^-(x_{-i})$ are monotone nondecreasing in $x_{-i}$.
\end{lemma}

We now move on with the next result characterizing properties of monotone and anti-monotone I- and BR-paths of super-modular games. 

\begin{lemma}\label{lemma:suprmodular-paths} For $x, y, z$ in $\mc X$, the following relations hold true: 
\begin{enumerate}
\item[(i)] $x\uparrow y, x\uparrow z\Rightarrow x\uparrow (y\vee z)$
\item[(ii)] $x\downarrow y, x\downarrow z\Rightarrow x\downarrow (y\wedge z)$
\item[(iii)] $x\uparrow y, x'\geq x\Rightarrow x'\uparrow (y\vee x')$
\item[(iv)] $x\downarrow y, x'\leq x \Rightarrow x'\downarrow (y\wedge x')$
\item[(v)] $x\to y \Rightarrow x\uparrow y',\; x\downarrow y''$ for some $y''\leq y\leq y'$.
\end{enumerate}
Moreover, the analogous results hold true for the BR-case. 
\end{lemma}
\begin{proof}
We prove (i). Let $(y^{(0)},y^{(1)},\ldots y^{(l)})$ and $(z^{(0)},z^{(1)},\ldots z^{(r)})$ be two monotone I-paths from $x$ to, respectively, $y$ and $z$. Let $(i_1,\dots ,i_l)$ and $(j_1,\dots , j_r)$ be the two corresponding sequences of active players. Let $(j_{s_1}, \dots ,j_{s_k})$ be the subsequence of $(j_1,\dots , j_r)$ consisting of exactly those players that are not in the sequence $(i_1,\dots ,i_l)$. We claim that the sequence $(x^{(0)},\dots , x^{(l+r)})$ defined by
\begin{itemize}
\item $x^{(h)}=y^{(h)}$ for $h=0,\dots , l$,
\item $x^{(h+l)}=x^{(h+l-1)}+\delta^{j_{s_h}}$ for $h=1,\dots ,k$
\end{itemize}
is a monotone I-path from $x$ to $y\vee z$. By construction, the path is admissible and monotone. Moreover,  $x^{(h+l-1)}\geq z^{(s_h-1)}$ for every $h=1,\dots , k$.  Since $\{+1\}=\mc B_{j_{s_h}}(z^{(s_h-1)})$, by Lemma \ref{lemma:basic-supermod}, $\{+1\}= \mc B_{j_{s_h}}(x^{(h-1+l)})$. This implies that it is an I-path. Proof of (ii) is completely analogous. 

We prove (iii). Let $(x^{(0)},x^{(1)},\ldots ,x^{(l)})$ be a monotone I-path from $x$ to $y$ with set of active players $(i_1,\dots ,i_l)$. Consider the subsequence $(i_{s_1}, \dots ,i_{s_k})$ of those players for which $x$ and $x'$ coincide. Then, $(x^{(0)}\vee x',x^{(i_{s_1})}\vee x',\ldots ,x^{(i_{s_k})}\vee x')$ is a monotone I-path from $x'$ to $y\vee x'$.  Indeed, notice that, by construction, $x^{(i_{s_k})}\vee x'=x^{(l)}\vee x'=y\vee x'$. We only need to show that it is an I-path. 
Since  
$(x^{(i_{s_h})}\vee x')_{-i_{s_h}}\geq x^{(i_{s_h})}_{-i_{s_h}}$ and using the increasing difference property (\ref{increasing-differences}) we obtain that
$$\begin{array}{rcl}0&\le&u_{i_{s_h}}(x^{(i_{s_h})})-u_{i_{s_{h}}}(x^{(i_{s_{h}-1})})\\&\le&u_{i_{s_h}}(x^{(i_{s_h})}\vee x')-u_{i_{s_{h}}}(x^{(i_{s_{h-1}})}\vee x')\,.\end{array}$$
The proof of (iii) is complete. The proof of (iv) is completely analogous. 

(v): If $x\uparrow y$, then $y\vee x=y$. If $x\downarrow y$, then $x\geq y$ and $y\vee x=x$. In both cases the result is evident. The general case can be proven by induction on the length of a minimal I-path from $x$ to $y$. Indeed, by definition of an I-path, for sure we can find an intermediate configuration $z$ for which one of the two possible cases hold: $x\uparrow z\to y$ or $x\downarrow z\to y$. In the first case, using the induction hypothesis $z\uparrow y'\geq y$, we obtain by transitivity that $x\uparrow y'\geq y$. In the second case, using the induction hypothesis $z\uparrow y'\geq y$ and point (iii), we obtain that $x\uparrow (x'\vee y')\geq y$. Similarly we prove the other relation.

Finally, the proofs for the analogous results in the BR-case can be obtained by the same identical arguments.
\end{proof}

The following result gathers some elementary key facts connected to the maps \eqref{4maps}. 

\begin{lemma}\label{lemma:f+-} The following facts hold true:
\begin{enumerate}
\item[(i)] $f^+$, $f^-$, $g^{+}$, $g^{-}$ are monotone nondecreasing maps;
\item[(ii)] a configuration $x$ in $\mc X$ is an equilibrium if and only if $$f^{+}(x)=x=f^{-}(x)\,;$$
\item[(iii)] a configuration $x$ in $\mc X$ is a strict equilibrium if and only if $$g^{+}(x)=x=g^{-}(x)\,.$$
\end{enumerate}
\end{lemma}

\begin{proof}
(i): It follows from Lemma \ref{lemma:suprmodular-paths} (i) that $x\uparrow f^+(x)$ for every configuration $x$ in $\mc X$. If $x'\geq x$, Lemma \ref{lemma:suprmodular-paths} (iii) yields $x'\uparrow f^+(x)\vee x'$. Therefore $f^+(x')\geq  f^+(x)\vee x'\geq  f^+(x)$. Proofs for 
$f^-$, $g^{+}$, $g^{-}$ are analogous. 

(ii) and (iii) coincide with the definitions of equilibrium and, respectively, strict equilibrium.
\end{proof}

We are now ready to prove Proposition \ref{prop:stability-supermodular}, which we restate below for the reader's convenience.

\setcounter{proposition}{1}
\begin{proposition} Consider a finite super-modular game with binary action sets. Then, for every configuration $x$ in $\mc X$,
\begin{enumerate}
\item[(i)]  $f^-(f^{+}(x))\in\mc X^*$ and $f^+(f^{-}(x))\in\mc X^*$ are, respectively, the greatest and least equilibria that are I-reachable from $x$;
\item[(ii)] $f^-(g^{+}(x))\in\mc X^*$ and $f^+(g^{-}(x))\in\mc X^*$ are, respectively, the greatest and least  equilibria BR-reachable from $x$; 
\item[(iii)]if $x\in\mc X^*$, then $g^{+}(x),g^{-}(x)\in\mc X^*$ are, respectively, the greatest and least equilibria BR-reachable from $x$. 
\end{enumerate}
Moreover,
\begin{enumerate}
\item[(iv)] the set of equilibria $\mc X^*$ is a complete lattice and is globally stable;
\item[(v)] the two configurations $$\begin{array}{l} \ul x^*=f^+(g^{-}(-\1))=f^+(-\1)\,,\\[5pt]\ov x^*=f^-(g^{+}(+\1))=f^-(+\1)\,,\end{array}$$
are, respectively, the least and greatest equilibria;
\item[(vi)] for two equilibria $x$ and $y$ in $\mc X^*$, $f^{+}(x\vee y)$ and $f^{-}(x\wedge y)$ are, respectively, the least and the greatest equilibria that are, respectively, above and below both $x$ and $y$;
\item[(vii)] if $\ul x^*= \ov x^*=x^*$, then  $\mc X^*=\{x^*\}$ is globally BR-stable.
\end{enumerate}
\end{proposition}
\begin{proof}
(i)  Put  $\mc X^+=\{x\in\mc X\,|\, f^+(x)=x\}$. We notice that for every $x\in\mc X$, $f^+(x)\in\mc X^+$. Moreover, $\mc X^+$ is closed with respect to anti-monotone I-path. Namely, if $x\in\mc X^+$ and $x\downarrow y$, then also $y\in\mc X^+$. To see this, by induction, it is sufficient to prove it when $x$ and $y$ are connected by an anti-monotone I-path of length $1$, namely there exists $i\in\mc V$ such that $y_i<x_i$, $y_{-i}=x_{-i}$, and $\mc B_i(y)=\{-1\}$. If $y\not\in \mc X^+$, then it would exists $z\in\mc X$ and a player $j\in\mc V$ such that $z_j>y_j$, $z_{-j}=y_{-j}$, and $\mc B_j(y)=\{+1\}$. Evidently $j\neq i$ and from Lemma \ref{lemma:suprmodular-paths} (iii) applied to $y\uparrow z$ and $x\geq y$ we would obtain $x\uparrow x\vee z$. Since by construction $x\vee z\neq x$, this would imply that $f^+(x)\neq x$ contrarily to the assumption that $x\in\mc X^+$. Similarly, $\mc X^-=\{x\in\mc X\,|\, f^-(x)=x\}$ is closed with respect to monotone I-path. Notice that $\mc X^*=\mc X^+\cap\mc X^-$.

Consider now $y=f^-(f^{+}(x))$. Being in the image of  $f^-$, necessarily $y\in\mc X^-$. On the other hand, since $f^+(x)\in\mc X^+$, by the fact that $\mc X^+$ is closed with respect to anti-monotone I-path, we have that also $y\in \mc X^+$. Hence $y$ is a Nash. The argument for $f^+(f^{-}(x))$ is completely analogous.

If now $y$ in $\mc X^*$ is any equilibrium reachable from $x$, namely $x\to y$, by Lemma \ref{lemma:suprmodular-paths} (v) it follows that $x\uparrow y'\geq y$. By definition of $f^+(x)$ we have that $f^+(x)\geq  y'\geq y$. Therefore, by Lemma \ref{lemma:f+-} (i), we {have} that 
$$f^-(f^{+}(x))\geq f^{-}(y)=y\,,$$
with the last equality above following from Lemma \ref{lemma:f+-} (ii). Similarly, we can show that $f^+(f^{-}(x))\leq y$. This concludes the proof of point (i).

Points (ii) and (iii) can be proven similarly point (i). We omit the details for the sake of conciseness.

(v) It follows from point (ii) that the configurations $\ul x^*$ and $\ov x^*$ defined in \eqref{extremalNash} are indeed equilibria. Given any equilibrium $x$ in $\mc X^*$, let $x^-=g^{-}(x)$ and $x^+=g^{+}(x)$. Point (iii) implies that both $x^-$ and $x^+$ are equilibria, hence 
$f^+(x^-)=x^-$ and $f^-(x^+)=x^+$ by Lemma \ref{lemma:f+-} (ii). Since $x^-\le x\le x^+$, it follows from Lemma \ref{lemma:f+-} (i) that
$$\ul x^*=f^+(-\1)\le f^+(g^{-}(x))=x^-\leq x \,,$$
$$x\le x^+=f^-(g^{+}(x))\le f^-(+\1)=\ov x^*\,,$$
so that $\ul x^*\le x\le\ov x^*$. 

(vi) For $x,y$ in $\mc X^*$ we have  $f^-(x)=x$  and $f^-(y)=y$ by Lemma \ref{lemma:f+-} (ii). It then follows from Lemma \ref{lemma:f+-} (i) that 
$$x=f^-(x)\le f^-(x\vee y)\,,\qquad y=f^-(y)\le f^-(x\vee y)\,,$$ so that
$$x\vee y\le f^-(x\vee y)\le x\vee y\,.$$
Hence, $f^-(x\vee y)=x\vee y$, so that $$f^{+}(x\vee y)=f^{+}(f^{-}(x\vee y))\in\mc X^*\,.$$ On the other hand, every $z$ in $\mc X^*$ such that $z\geq x$ and $z\geq y$ is such that $z\geq x\vee y$ and thus $z=f^+(z)\geq f^{+}(x\vee y)$ proving that $f^{+}(x\vee y)$ is the smallest of the Nash above both $x$ and $y$. The claim on $f^{-}(x\wedge y)$ is completely analogous.

(iv) Point (vi) implies that the set of equilibria $\mc X^*$ is a complete lattice. By definition, $\mc X^*$ is invariant, while  global reachability follows from point (i).

(vii) Point (iii) that implies $\mc X^*=\{x^*\}$ is globally I-reachable and thus also globally BR-reachable. Moreover, notice that point (iii) implies that both $g^{-}(x^*)$ and $g^{+}(x^*)$ are equilibria, so that the assumption $\mc X^*=\{x^*\}$ implies that
$g^{-}(x^*)=x^*=g^{+}(x^*)$. Hence $\{x^*\}$ is BR-invariant.  
\end{proof}

\section{Proof of Lemma \ref{lemma:coordination}}\label{sec:proof-lemma-coordination} 
(i)  This follows directly from the equivalent form \eqref{eq:utility2} of the utility function. 

(ii) By substituting the identity $w_i=w_i^+(x)+w_i^-(x)$ into \eqref{eq:utility2}, we have that 
$$u_i(x)=x_i\left(h_i+2w_i^+(x)-w_i\right)\,,$$ 
from which \eqref{eq:BR} follows directly. 

(iii) This follows directly from point (ii). 

(iv) It immediately follows from the computation
$$u_i(1,x_{-i})-u_i(-1,x_{-i})=2\sum\nolimits_{j}W_{ij}x_j+2h_i\,,$$
and nonnegativity of the link weights $W_{ij}$ that  a network coordination game is super-modular. 

(v) It follows from point (iv) that the coordination game on $\mc G$ with external field $h$ is super-modular. Then, Proposition \ref{prop:stability-supermodular} (iv) implies the claim. \qed

\section{Proof of Lemma \ref{lemma:coordination=LTM}}\label{sec:proof-lemma-coordination=LTM} 
(i) It follows from Lemma \ref{lemma:coordination} (i), the form of the utility functions \eqref{utility:coordination} of the coordination game, and that of the update rule \eqref{dynamics} of the LTD, that $X(t)$ can have a transition from a configuration $x$ in $\mc X$ to another configuration $y$ in $\mc X$ if and only if there exists $i$ in $\mc V$ such that $x_{-i}=y_{-i}$, $y_i\ne x_i$, and $\mc B_i(y_{-i})=\{y_i\}$. It follows that $X(t)$ can undergo a finite sequence of transitions from $ x^{(0)}$ to $x^*$ if and only if $ x^{(0)}\stackrel{h}{\to}x^*$. 

(ii)  It follows from Lemma \ref{lemma:coordination} (i) and the form of the utility functions \eqref{utility:coordination} of the coordination game,  that $x_i^{\circ}\in\mc B_i( x^{(0)}_{-i})$ if and only if 
\be\label{xiBi}x_i^{\circ}\left(\sum\nolimits_{j}W_{ij}x_j^\circ+h_i\right)\ge0\,.\ee
Hence $ x^{(0)}\in\nash$ if and only if \eqref{xiBi} holds true for every $i$ in $\mc V$.
The form of the update rule \eqref{dynamics} of the LTD implies that this is the case if and only if with probability $1$ $X(t)= x^{(0)}$ for every $t\ge0$. \qed

\section{Proof of Lemma \ref{lemma:consensus}}
\label{sec:proof-lemma-consensus}
If $\ul x^*(h)=-\1$, then  there cannot be $+1$-stubborn players, i.e., $\mc S_{+1}(h)=\emptyset$, so that $h\le w$. On the other hand, if $h\le w$, then, using the threshold form of the best response in Lemma \ref{lemma:coordination} (ii), we deduce that $-\1$ is an equilibrium, namely $\ul x^*(h)=-\1$. This proves (i), while (ii) can be proven analogously.\qed

\section{Proof of Proposition \ref{prop:polarization}} \label{sec:proof-prop-polarization} 
Given any configuration $x^*$ in $\mc X$, for any player $i$ such that $x^*_i=s$,  from  \eqref{eq:utility2}  we can write that
\be\label{utility}\ba{rcl}
u_i(x^*)
&=&s(h_i+w_i^+(x^*)-w_i^{-}(x^*))\\
&=&s(h_i+sw_i^s(x^*)-sw_i^{-s}(x^*))\\
&=&sh_i+w_i^s(x^*)-w_i^{-s}(x^*)\,.
\ea\ee
We now argue as follows.
If the network coordination game on $\mc G$ with external field $h$ is polarizable, then there exists an equilibrium $x^*\ne\pm\1$. From (\ref{utility}) and Lemma \ref{lemma:coordination} (i) we derive that, for every $s$ and $i$ such that $x^*_i=s$,
$$sh_i+w_i^s(x^*)-w_i^{-s}(x^*)\geq 0\,.$$
This implies that relatively to the nontrivial binary partition $\mc V=\mc V_{x^*}^+\cup\mc V_{x^*}^-$, \eqref{eq:robustly_indec} is violated for every $i$ in $\mc V_{x^*}^s$ and $s$ in $\{\pm\}$. 
Hence, $\mc G$ is not $h$-indecomposable. 

On the other hand, if $\mc G$ is not $h$-indecomposable, then by Proposition \ref{prop:indecomposable-necessary} there exists a co-existent absorbing configuration $x^*$of the LTD on $\mc G$ with constant external field $h$. By Lemma \ref{lemma:coordination=LTM} (ii), $x^*$ in $\nashcoexistent$ is a coexistent equilibrium of the  network coordination game on $\mc G$ with external field $h$, which is then polarizable. \qed

\section{Proof of Proposition \ref{prop:biased-indecomposable}}\label{sec:proof-prop-biased-indecomposable} 
(Only if) Assume that the network coordination game is $a$-biased and unpolarizable. Then, condition (a) follows from Proposition \green{\ref{prop:regular-biased-frustrated}} (ii). To prove (b), assume by contradiction that there exists a non-empty subset $\mc R\subseteq\mc V\setminus\mc S_a(h)$ such that 
\be\label{cond-biased-indecomp}w^{\mc R}_i\ge w^{\mc V\setminus\mc R}_i+ah_i\,,\ee
for every $i$ in $\mc R$ and let $x$ in $\mc X$ be a configuration such that $x_i=a$ for every $i$ in $\mc V\setminus\mc R$ and $x_i=-a$ for every $i$ in $\mc R$.
Notice that 
\eqref{cond-biased-indecomp} and \eqref{utility} imply that  $u_i(x)\ge 0$, so that, by Lemma \ref{lemma:coordination} (i), $-a=x_i\in\mc B_i(x_{-i})$, for every $i$ in $\mc R=\mc V^{-a}_{x}$. If $a=+1$ ($a=-1$), this implies that there are no monotone (anti-monotone) I-paths of positive length starting at $x$, so that in particular 
$f^a(x)=x$. 
Then, by Proposition \ref{prop:stability-supermodular} (i), we get that 
$$x^*=f^{-a}(x)=f^{-a}(f^a(x))$$ 
is an equilibrium. Now, ntice that on the one hand $x^*_i=-a$ for every $i$ in $\mc R$ (since $x_i=-a$ and $x^*=f^{-a}(x)$), on the other hand $x^*_i=a$ for every $i$ in $\mc S_a(h)$ (since those are stubborn players). 
Hence, $x^*$ is a co-existent equilibrium, thus contradicting the assumption that the game is unpolarizable. Therefore, if the network coordination game is  $a$-biased and unpolarizable, both conditions (a) and (b) must be satisfied. 

(If) Given any player $i$ in $\mc V\setminus\mc S_a(h)$, from the application of \eqref{indecomposable-simple} with $\mc R=\{i\}$ we obtain that 
$$ah_i+ w_i\ge a h_i+ w_i^{\mc V\setminus\mc R}>w_i^{\mc R}\ge0\,. $$
Since $ah_i+ w_i>ah_i-w_i\ge0$ for every $i$ in $\mc S_a(h)$, we deduce that $ah+w\ge 0$. Together with assumption (a), by Proposition \ref{prop:regular-biased-frustrated} (ii), this yields that the game is $a$-biased. 
We finally prove that the game is unpolarizable. By contradiction, suppose there exists a co-existent equilibrium $x^*$. Necessarily $x^*_i=a$ for every $i$ in $\mc S_a(h)$. Put $\mc R=\mc V_{x^*}^{-a}$ and notice that (\ref{utility}) yields
$$\ba{rcl}0 &\leq& u_i(x^*)\\
&=&-ah_i+w_i^{-a}(x^*)-w_i^{a}(x^*)\\
&=&-ah_i+w_i^{\mc R}-w_i^{\mc V\setminus\mc R}\ea
$$
for every $i$ in $\mc R$ thus contradicting condition (b). The proof is then complete.\qed

\bibliographystyle{ieeetr}
\bibliography{bib,bib2}

\begin{IEEEbiography}
[{\includegraphics[width=1in,height=1.25in,clip,keepaspectratio]{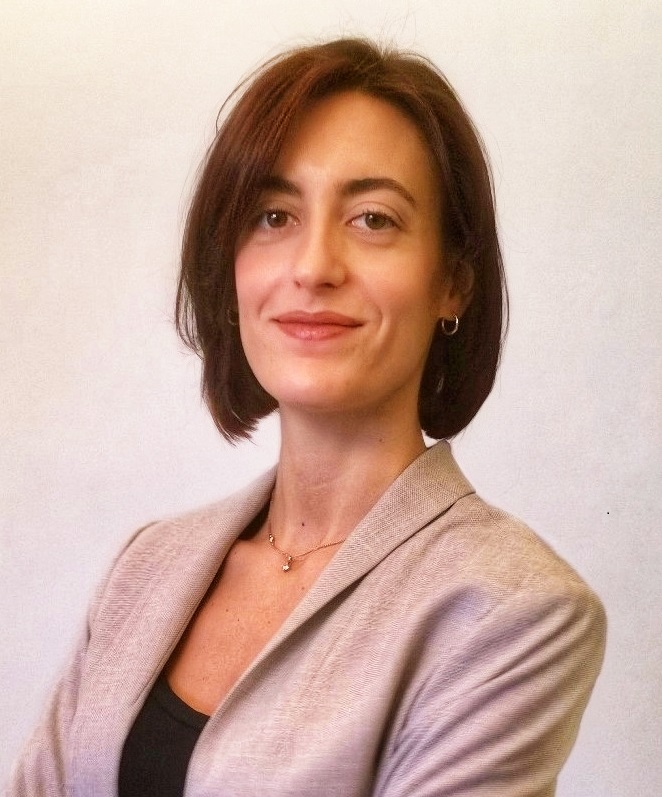}}]
{Laura Arditti} is a PhD student in Applied Mathematics at the Department  of  Mathematical  Sciences,  Politecnico di  Torino,  Italy. 
She  received the B.Sc. in Physics Engineering in 2016 and the M.S. in Mathematical Engineering in 2018, both {\it magna cum laude} from Politecnico di Torino. Her research is concentrated on game theory, its relationship with graphical models, and its applications to infrastructure, social, economic, and financial networks.

\end{IEEEbiography}

\begin{IEEEbiography}
[{\includegraphics[width=1in,height=1.25in,clip,keepaspectratio]{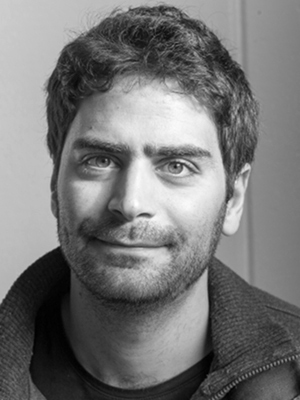}}]
{Giacomo Como}
is  a  Professor at  the Department  of  Mathematical  Sciences,  Politecnico di  Torino,  Italy,  and a Senior Lecturer  at  the  Automatic  Control  Department  of  Lund  University,  Sweden.  He  received the B.Sc., M.S., and Ph.D. degrees in Applied Mathematics  from  Politecnico  di  Torino,  in  2002,  2004, and 2008, respectively. He was a Visiting Assistant in  Research  at  Yale  University  in  2006--2007  and  a Postdoctoral  Associate  at  the  Laboratory  for  Information  and  Decision  Systems,  Massachusetts  Institute of Technology, from 2008 to 2011. 
He currently serves as Senior Editor  of the \textit{IEEE Transactions on Control of Network Systems}, as Associate Editor of \textit{Automatica}  and  as  chair  of the  {IEEE-CSS  Technical  Committee  on  Networks  and  Communications}. 
He has been serving as  Associate  Editor  of the  \textit{IEEE Transactions on Network Science and Engineering} and of the \textit{IEEE Transactions on Control of Network Systems}. 
 He was  the  IPC  chair  of  the  IFAC  Workshop  NecSys'15  and  a  semiplenary speaker  at  the  International  Symposium  MTNS'16.  He  is  recipient  of  the 2015  George S. ~Axelby  Outstanding Paper Award.  His  research interests  are in  dynamics,  information,  and  control  in  network  systems  with  applications to  cyber-physical  systems,  infrastructure  networks,  and  social  and  economic networks.
\end{IEEEbiography}

\begin{IEEEbiography}
[{\includegraphics[width=1in,height=1.25in,clip,keepaspectratio]{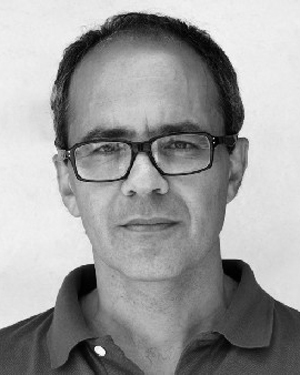}}]
{Fabio Fagnani}
received the Laurea degree in Mathematics from the University of Pisa and the Scuola Normale Superiore, Pisa, Italy, in 1986. He received the PhD degree in Mathematics from the University of Groningen,  Groningen,  The  Netherlands,  in 1991. From 1991 to 1998, he was an Assistant Professor of Mathematical Analysis at the Scuola Normale Superiore. In 1997, he was a Visiting Professor at the Massachusetts Institute of Technology (MIT), Cambridge, MA. Since 1998, he has been with the Politecnico of Torino, where since 2002 he has been a Full Professor of Mathematical Analysis. From 2006 to 2012, he has acted as Coordinator of the PhD program in Mathematics for Engineering Sciences at Politecnico di Torino. From June 2012 to September 2019, he served as the Head of the Department of Mathematical Sciences, Politecnico di Torino. His current research topics are on cooperative algorithms and dynamical systems over graphs, inferential distributed algorithms, and opinion dynamics. He is an Associate Editor of the \textit{IEEE Transactions on Automatic Control} and served in the same role for the \textit{IEEE Transactions on Network Science and Engineering} and of the \textit{IEEE Transactions on Control of Network Systems}.
\end{IEEEbiography}

\begin{IEEEbiography}
[{\includegraphics[width=1in,height=1.25in,clip,keepaspectratio]{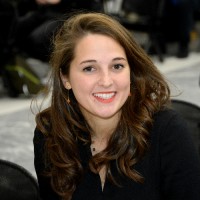}}]
{Martina Vanelli} is a PhD student in Applied Mathematics at the Department  of  Mathematical  Sciences (DISMA),  Politecnico di  Torino,  Italy. 
She  received the B.Sc. and M.S. degrees in Applied Mathematics  from  Politecnico  di  Torino,  in  2017 and  2019, respectively. From October 2018 to March 2019, she was a visiting student at  Technion, Israel. Her research interests include game theory, auction theory, and network systems with applications to social and economic networks and power markets.

\end{IEEEbiography}

\end{document}